\documentclass[11pt]{amsart}

\usepackage[lite]{amsrefs}
\usepackage{amssymb}
\usepackage[all,cmtip]{xy}
\usepackage{amsthm}
\usepackage{amsmath}
\usepackage{amsfonts}
\usepackage[hidelinks]{hyperref}
\usepackage{enumerate}
\usepackage{physics}
\usepackage[titletoc]{appendix}
\usepackage{fancyhdr}
\usepackage{xcolor}
\usepackage{comment}
\usepackage{tikz-cd}
\usepackage{tikzpagenodes}
\usepackage{caption}
\usepackage{float}
\usetikzlibrary[patterns]
\usepackage[margin = 1in]{geometry}
\usepackage{setspace}

\spacing{1.2}

\allowdisplaybreaks[3]

\numberwithin{equation}{section}
\newtheorem{definition}{Definition}[section]
\newtheorem{example}[definition]{Example}
\newtheorem{theorem}[definition]{Theorem}
\newtheorem{lemma}[definition]{Lemma}
\newtheorem{corollary}[definition]{Corollary}
\newtheorem{proposition}[definition]{Proposition}
\newtheorem{remark}[definition]{Remark}

\newtheorem{pdefinition}[definition]{Proposition-Definition}

\usepackage{textcomp}
\usepackage{graphicx}
\usepackage{amssymb,amsmath,amsthm,amscd,bm}

\usepackage{mathrsfs}
\usepackage{lscape}
\usepackage{enumerate}
\usepackage{tikz}
\usepackage[all]{xy}
\usetikzlibrary{matrix}

\usepackage{mathtools}

\usepackage{comment}
\usepackage{amsmath}
\newcommand{\Mod}[1]{\ (\mathrm{mod}\ #1)}

\usepackage{color}

\numberwithin{equation}{section}

\newcommand{\ber}{\begin{red}}
\newcommand{\er}{\end{red}}
\newcommand{\beb}{\begin{blue}}
\newcommand{\eb}{\end{blue}}

\newcommand\atopn[2]{\genfrac{}{}{0pt}{2}{#1}{#2}}

\theoremstyle{plain}

\newcommand{\n}{\mathfrak{n}}
\newcommand{\pp}{\mathfrak{p}}
\newcommand{\g}{\mathfrak{g}}
\newcommand{\RR}{\mathbb{R}}
\newcommand{\CC}{\mathbb{C}}

\newcommand{\ZZ}{\mathbb{Z}}
\newcommand{\NN}{\mathbb{N}}

\newcommand{\al}{\alpha}

\newcommand{\be}{\beta}

\newcommand{\sll}{\mathfrak{sl}}
\newcommand{\osp}{\mathfrak{osp}}
\newcommand{\ad}{\text{ad}}
\newcommand{\WW}{\mathcal{W}}

\newcommand{\w}{\omega}

\newlength{\mylength}
\def\fk{\mathfrak{h}}
\def\sn{\mathsf{n}}

\def\cI{\mathcal{I}}
\def\cJ{\mathcal{J}}

\def\fh{H}

\def\ri{\rm I}

\title[$N=2$ SUSY structures on
classical $W$-algebras]{$N=2$ supersymmetric structures\\ on
classical $W$-algebras}

\author
[E. Ragoucy]{Eric Ragoucy}
\address[E. Ragoucy]{Laboratoire de Physique Th\'{e}orique LAPTh,
CNRS, Universit\'{e} Savoie Mont Blanc,
BP 110, 74941 Annecy-le-Vieux Cedex, France}
\email{eric.ragoucy@lapth.cnrs.fr}
\author
[A. Song]{Arim Song$^{1,2}$}
\address[A. Song]{Department of Mathematical Sciences, Seoul National University, GwanAkRo 1, Gwanak-Gu, Seoul 08826, Korea}
\email{ireansong@snu.ac.kr}
\author[U.R. Suh]{Uhi Rinn Suh$^{1}$}
\address[U.R. Suh]
{ Department of Mathematical Sciences and Research institute of Mathematics, Seoul National University, GwanAkRo 1, Gwanak-Gu, Seoul 08826,
Korea}
\email{uhrisu1@snu.ac.kr}

\thanks{$^{1}$This work was supported by NRF Grant,  \#2022R1C1C1008698 and Creative-Pioneering Researchers Program through Seoul National University}
\thanks{$^{2}$This research was supported by Basic Science Research Program through the National Research Foundation of Korea(NRF) funded by the Ministry of Education(RS-2023-00272036)}

\date{\today}

\begin{document}

\begin{abstract}
We describe a $N=2$ supersymmetric Poisson vertex algebra structure of $N=1$ (resp. $N=0$) classical $W$-algebra associated with $\mathfrak{sl}(n+1|n)$ and the odd (resp. even) principal nilpotent element. This $N=2$ supersymmetric structure is connected to the 
principal $\mathfrak{sl}(2|1)$-embedding in $\mathfrak{sl}(n+1|n)$ superalgebras, which are the only basic Lie superalgebras that admit such a principal embedding.
\end{abstract}

\maketitle

\section{Introduction} \label{sec: intro}
Supersymmetric (SUSY) vertex algebras were introduced by Kac \cite{Kac98}, 
generalizing the original idea of Borcherds \cite{Borch} for (bosonic) vertex algebras. 
Afterwards, SUSY vertex algebras have been studied in various aspects, see e.g. \cites{Bar00,Hel07,BHS08}, and
their structure theory was developed by Kac and Heluani \cite{HK07}.
SUSY vertex algebras correspond to a mathematical 
definition of the chiral part of two-dimensional (super)conformal field theories studied by physicists, 
since the years 
80's  up  to nowadays, see e.g. \cites{BPZ,W84,KZ,GW,BS,Za,FMS,Pope,BS} and references therein for the historical physics literature on the subject. 
Their structure of SUSY vertex algebras includes a superfield formalism,  commonly used in physics, for instance in 
superstrings theories \cites{CHSW,HW} or super-Yang-Mills models \cites{gaiotto, BMN,DHKS, BDKRSVV}.

The affine Lie algebras and the $W$-algebras are at the core of the study of these SUSY theories, see e.g.  \cites{Ademollo, DRS92,MadR,IMY}. However, although supersymmetric vertex algebras 
may have as many supersymmetries as one wishes, 
it is mainly the case  $N=1$  that is studied in the context of affine Lie algebras \cite{KRW04} and $W$-algebras \cites{MRS21,Suh20}. 
Affine Lie algebras, and $W$-algebras are studied both at the quantum level (using vertex algebra and 
SUSY vertex algebra formalism) and at the classical level (with Poisson vertex algebras and SUSY Poisson 
vertex algebras). 

Quantum and classical SUSY W-algebras can be understood as Hamiltonian reductions of $N=1$ SUSY affine vertex algebras in the quantum case or SUSY Poisson affine vertex algebras in the classical case.
As a method of $N=1$ SUSY Hamiltonian reduction, BRST formalism was introduced  
and this formalism gives rise to $N=1$ SUSY quantum W-algebras \cites{MadR,MRS21}.
On the other hand, SUSY classical W-algebras are obtained as a quasi classical limit of SUSY quantum W-algebras \cite{Suh20}.
A $N=1$ SUSY classical (resp. quantum) W-algebra $\mathcal{W}^k(\bar{\mathfrak{g}},f)$ (resp. $W^k(\bar{\mathfrak{g}},f)$) of level $k$ is defined when a basic Lie superalgebra $\mathfrak{g}$ and an odd nilpotent element $f$ in an $\mathfrak{osp}(1|2)$-subalgebra of $\mathfrak{g}$ are given.
Especially, for classical cases, their Poisson structures are  well understood \cite{Suh20} and when the given odd nilpotent element $f$ is principal, the explicit forms of their generators can be found by a column determinant  of certain operator valued matrices \cite{RSS23}.

The $N=2$ SUSY $W$-algebras have been used since a long time in physics. They represent a significant extension of the already rich N=2 superconformal symmetry.  The algebraic structure of $N=2$ superconformal algebra is particularly fitted to the study of strings  and superstrings. Indeed, after a well-suited twist of the super conformal generators, it can be shown \cite{GS92} that the $N=2$ superconformal algebra encode the BRST algebra for the bosonic string. In that approach, the BRST algebra contains, in addition to the stress energy tensor, the BRST operator, together with a ghost and an anti-ghost operators. Extension of this method to the bosonic $W$-strings leads naturally to the $N=2$ SUSY $W$-algebras \cite{BLN93}, which represent a significant extension of the already rich $N=2$ superconformal symmetry. These algebras also appear as the Hamiltonian structure of the $N=2$ Korteweg--de Vries or of the $N=2$
Kadomtsev--Petriashvili hierarchies \cites{LM88, DG97}.  These various lines of research have ignited a multitude of studies on these algebras, with a particular focus on the physical perspective \cites{RSS96, FOFR,delius}. 

However, from the mathematical side, the (Poisson) vertex algebra counterpart of the $N=2$ SUSY $W$-algebras remains elusive, particularly within the framework of $N=2$ superfields. The objective of this paper is to bridge this gap concerning $N=2$ SUSY structures of classical W-algebras. In particular, we describe the $N=2$ SUSY structures of the non-SUSY classical W-algebra $\mathcal{W}^k(\mathfrak{sl}(n+1|n),F)$ and the $N=1$ SUSY classical W-algebra $\mathcal{W}^k(\overline{\mathfrak{sl}}(n+1|n),f)$ associated with the Lie superalgebra $\mathfrak{sl}(n+1|n)$. Here $F$ is the even principal nilpotent and $f$ is the odd principal nilpotent. The reason we consider $\mathfrak{sl}(n+1|n)$ is that it is the only basic Lie superalgebra which has principal $\mathfrak{sl}(2|1)$-embedding \cite{DRS92}, \cite{RSS96}. It is however expected that the even (resp. odd) nilpotent element in any subalgebra isomorphic to $\mathfrak{sl}(2|1)$  yields a $N=2$ SUSY structure of $\mathcal{W}^k(\mathfrak{sl}(n+1|n),F)$  (resp. $\mathcal{W}^k(\overline{\mathfrak{sl}}(n+1|n),f)$) \cite{RSS96}.

In order to find $N=2$ SUSY structure of a Poisson vertex algebra, it is enough to find a $N=2$ superconformal vector \cite{HK07}. Inspired from this fact, we find a $N=2$ superconformal vector of $\mathcal{W}^k(\overline{\mathfrak{sl}}(n+1|n), f)$ and $N=1$ and $N=2$ superconformal vector of $\mathcal{W}^k(\mathfrak{sl}(n+1|n), F)$.
More precisely, as main results of this paper, we prove the following statements (Theorem \ref{thm:N=1 to N=2, sl(n+1|n)} and Theorem \ref{thm:N=2 conformal,N=0 W-alg}):
\begin{itemize}
    \item  $\mathcal{W}^k(\mathfrak{sl}(n+1|n),F)$ has $N=1,2$ SUSY structures. 
    \item There is a $N=1$ (resp. $N=2$) superconformal vector $G$ (resp. $J$) of $\mathcal{W}^k(\mathfrak{sl}(n+1|n),F)$ and a generating set as a $N=1$ (resp. $N=1$) SUSY Poisson vertex algebra consisting of $G$-primary (resp. $J$-primary) elements.
     \item  $\mathcal{W}^k(\overline{\mathfrak{sl}}(n+1|n),f)$ has $N=1,2$ SUSY structures. 
    \item There is a $N=1$ (resp. $N=2$) superconformal vector $G$ (resp. $J$) of $\mathcal{W}^k(\overline{\mathfrak{sl}}(n+1|n),f)$ and a generating set as a $N=1$ (resp. $N=2$) SUSY Poisson vertex algebra consisting of $G$-primary (resp. $J$-primary) elements.
\end{itemize}
In the proof of the listed results, we crucially use explicit formulas for $\lambda$-brackets of the classical W-algebras which appear in \cite{Suh20}. Here the $\lambda$-bracket is a way to express 
the OPE relation between two elements of a Poisson vertex algebra.

The plan of the paper is the following. Some general properties of $\mathfrak{sl}(n+1|n)$ with respect to their 
$\osp(1|2)$ and $\mathfrak{sl}(2|1)$ principal embeddings are recalled in section \ref{sec: sl(2|1) repn}.
In section \ref{sec: PVA and SUSY PVA}, we review the definition and some properties of supersymmetric Poisson vertex algebras. 
Then, in section \ref{sec:W-algebra}, the structure of $N=0,1$ SUSY $W$-algebras is studied.
Section \ref{sec: additional SUSY structure on W-algebras} deals with $N=2$ SUSY $W$-algebras and contains
the main results of the paper. We conclude in section \ref{sec:conclu} on open problems. An appendix gathers the technical proofs of the results presented in section \ref{sec: additional SUSY structure on W-algebras}.

\section{\texorpdfstring{$\mathfrak{sl}(n+1|n)$}{}  as \texorpdfstring{$\osp(1|2)$}{}  and \texorpdfstring{$\mathfrak{sl}(2|1)$}{}-module\label{sec: sl(2|1) repn}}
It has been observed in \cite{RSS96} that we should consider basic Lie superalgebras with $\mathfrak{sl}(2|1)$ embedding to possibly get $N_K=2$ supersymmetry structure in $W$-algebras. Therefore, a Lie superalgebra with principal $\mathfrak{sl}(2|1)$-embedding would be the first thing to study for $N_K=2$ supersymmetry.

\begin{proposition}[\cite{DRS92}, \cite{RSS96}]
    The only basic Lie superalgebras that admit a principal $\mathfrak{sl}(2|1)$-embedding are $\mathfrak{sl}(n+1|n)$ for $n\in \NN$.
    
    We remind that the basic (or contragredient) Lie superalgebras include the $A(m,n)=\mathfrak{sl}(n+1|n)$ unitary series,
the $B(m,n)=\osp(2m+1|2n)$, $C(n+1)=\osp(2|2n)$, and $D(m,n)=\osp(2m|2n)$ orthosymplectic series and the exceptional
superalgebras $F(4)$ and $G(3)$ as well as $D(2,1;\alpha)$.
\end{proposition}
\begin{proof}
  Assume that $\g$ is a basic Lie superalgebra that admits a principal $\mathfrak{sl}(2|1)$-embedding. Since $\g$ also has a principal $\mathfrak{osp}(1|2)$-embedding, $\g$ should be isomorphic to one of the followings:
  \begin{equation*}
    \mathfrak{sl}(n\!\pm\!1|n),\hskip 2mm \mathfrak{osp}(M|2n) \text{ with } M=2n\!\pm\! 1, 2n\!+\!2, 2n,\hskip 2mm
    D(2,1;\alpha) \text{ with }\alpha\!\in \CC\!\setminus\!\{0, \pm 1\}.
  \end{equation*}
   Each type of the above Lie superalgebras has an irreducible $\mathfrak{osp}(1|2)$-module decomposition
   \begin{align*}
    \mathfrak{sl}(n\!-\!1|n)&\simeq R_{2n-2}\oplus R_{2n-3}\oplus R_{2n-4}\oplus \cdots R_{2}\oplus R_{1},\\
    \mathfrak{sl}(n\!+\!1|n)&\simeq R_{2n}\oplus R_{2n-1}\oplus R_{2n-2}\oplus \cdots R_{2}\oplus R_{1},\\
    \mathfrak{osp}(2n|2n)&\simeq R_{4n-2} \oplus R_{4n-5} \oplus R_{4n-6} \oplus \cdots \oplus R_{3} \oplus R_{2} \oplus R_{2n-1},\\
    \mathfrak{osp}(2n\!-\!1|2n)&\simeq R_{4n-2} \oplus R_{4n-5} \oplus R_{4n-6} \oplus \cdots \oplus R_{3} \oplus R_{2},\\
    \mathfrak{osp}(2n\!+\!1|2n)&\simeq R_{4n-1} \oplus R_{4n-2} \oplus R_{4n-5} \oplus \cdots \oplus R_{3} \oplus R_{2},\\
    D(2,1;\alpha)&\simeq 2R_{2}\oplus R_{3},
   \end{align*}
   where $R_q$ is the finite irreducible $\mathfrak{osp}(1|2)$-module of dimension $2q+1$. Since $\mathfrak{sl}(2|1)\simeq R_{2}\oplus R_{1}$ as an $\mathfrak{osp}(1|2)$-module, $\g$ should contain $R_{2}\oplus R_{1}$, which is true only for $\mathfrak{sl}(n+1|n)$, $\mathfrak{sl}(n|n+1)\simeq \mathfrak{sl}(n+1|n)$ and $\mathfrak{osp}(2|2)\simeq \mathfrak{sl}(2|1)$. In these cases, one can check that the subspace $R_{2}\oplus R_{1}$ is isomorphic to $\mathfrak{sl}(2|1)$ as Lie superalgebras.
\end{proof}

Considering the above proposition, this paper only deals with $W$-algebras associated with $\mathfrak{sl}(n+1|n)$, that we will denote $\g$ throughout 
this section.
As a start, we remind the module structure of $\mathfrak{sl}(n+1|n)$ as an adjoint representation of the principal $\mathfrak{sl}(2|1)$ subalgebra.

\subsection{\texorpdfstring{$\mathfrak{sl}(n+1|n)$} as an \texorpdfstring{$\mathfrak{osp}(1|2)$}-module} \label{sec: osp(1|2)-repn}
We first review the structure of $\g$ as an $\mathfrak{osp}(1|2)$-module given by the adjoint representation of the principal $\mathfrak{osp}(1|2)$ subalgebra.

Recall that $\mathfrak{osp}(1|2)$ has a basis $\{\fh ,E,F,e,f\}$ satisfying the following properties:
\begin{enumerate}[]
  \item(osp-1) $\fh $ is a nontrivial element in the Cartan subalgebra of $\mathfrak{osp}(1|2)$,
  \item(osp-2) $E,F$ are even eigenvectors of $\textup{ad}\fh $ with eigenvalues $1,-1$,
  \item(osp-3) $e,f$ are odd eigenvectors of $\textup{ad}\fh $ with eigenvalues $\frac12,-\frac12$,
  \item(osp-4) $[e, e]=2 E,\hskip 2mm [f, f]=-2 F,\hskip 2mm [e, f]=-2\fh ,\hskip 2mm [F, e]=f,\hskip 2mm [E, f]=e$,
\end{enumerate}
where $[.,.]$ is the graded commutator, which is graded skewsymmetric and obeys the graded Jacobi identity.
We denote by $\fh ,E,F,e,f \in \g$ the elements in the principal $\osp(1|2)$ subalgebra of $\g$ satisfying (osp-1)--(osp-4).
With respect to the $\osp(1|2)$ embedding of $\g$, we define a gradation of $\g$:  
\begin{equation} \label{sl(n+1|n) eigenspace decomposition}
  \g=\bigoplus_{i\in \ZZ/2}\g(i), \text{ where  } \g(i)=\{x\in \g| [\fh,x]=i\, x\}.
\end{equation}

Recall that an irreducible $\osp(1|2)$-module can be identified by its dimension which is odd. Let $R_i$ for $i\in \ZZ_+$ be the irreducible module of $\osp(1|2)$ of dimension $2i+1$. Then the following proposition holds (see, for example \cite{Wak99}).

\begin{proposition} \label{prop:osp repn property}
  Consider $\g$ as an $\mathfrak{osp}(1|2)$-representation given by the adjoint action of the principal $\mathfrak{osp}(1|2)$ subalgebra. Then $\g$ decomposes into the direct sum of irreducible representations:
 \begin{equation} \label{eq: osp(1|2) irreducible decomposition}
    \g=\bigoplus_{i=1}^{2n} R_i.
 \end{equation}
\end{proposition}

The irreducible component $R_i$ is decomposed into a sum of 1-dimensional subspaces $R_{i}(j):=R_i \cap \g(j)$ for $j\in \frac{1}{2}\ZZ$ such that $-\frac{i}{2} \leq j \leq \frac{i}{2}$. Hence, if we take a highest weight vector $v_i\in R_i$, then 
\begin{equation} \label{eq:fix element in V_ij}
v^{(m)}_{i}:=(\textup{ad}f)^{m}v_i
\end{equation}
spans $R_{i}(\frac{i-m}{2})$, $0\leq m\leq2i$, and $v^{(0)}_{i}=v_i$ is the highest weight vector. In addition, we have bases $\{v_i^{(0)}|i=1,2,\cdots, 2n\}$ and $\{v^{(2i)}_i|i=1,2,\cdots, 2n\}$
of $\g^e:=\text{ker}(\textup{ad}e)$ and $\g^f:=\text{ker}(\textup{ad}f)$, respectively.

\subsection{$\mathfrak{sl}(n+1|n)$ as an $\mathfrak{sl}(2|1)$-module} \label{sec: sl(2|1) in detail}
 In this section, we consider the principal $\mathfrak{sl}(2|1)$-embedding in $\g$. For the principal $\mathfrak{sl}(2|1)$ subalgebra, fix a basis $\{\fh ,E,F,e,f,\tilde{e}, \tilde{f}, U\}$ satisfying the following properties:

\begin{enumerate}[]
  \item (sl-1) $\fh ,E,F,e,f$ have the properties (osp-1)--(osp-4),
  \item (sl-2) $\fh ,-E,-F, \tilde{e}, \tilde{f}$ also obey (osp-1)--(osp-4),
  \item (sl-3) $[f, \tilde{f}]=[e, \tilde{e}]=0$,
  \item (sl-4) $[e, \tilde{f}]=[\tilde{e}, f]=U$, $[U, f]=-\tilde{f}$, $[U, \tilde{f}]=-f$, $[U, e]=\tilde{e}$,  $[U, \tilde{e}]=e$. 
\end{enumerate}
\vskip 2mm 

To describe the adjoint representation of the principal $\mathfrak{sl}(2|1)$ subalgebra, we use the gradation \eqref{sl(n+1|n) eigenspace decomposition} and the irreducible component $R_i$ in Proposition \ref{prop:osp repn property}.

\begin{lemma} \label{lemma:sl repn}
  Let $a\in \g$ be any homogeneous element in the irreducible component $R_i$ satisfying $[U, a]=0.$
  Then the following properties hold.
  \begin{enumerate}
    \item  Both $[E, a]$ and $[F,a]$ are killed by $\textup{ad}U$.
    \item  The following diagrams commute:
  
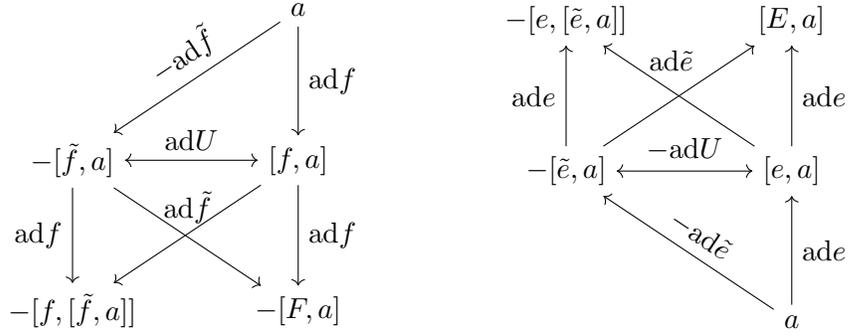
\begin{center}
  \begin{tikzpicture}
    \node (ftildv) at (-2,1) {$-[\tilde{f},a]$};
    \node (ftildfv) at (-2,-1) {$-[f,[\tilde{f},a]]$};
    \node (v) at (1,3) {$a$};
    \node (fv) at (1,1) {$[f,a]$};
    \node (ffv) at (1,-1) {$-[F,a]$};
    \draw[->] (v) -- (fv) node[midway, right] {$\textup{ad}f$};
    \draw[->] (fv) -- (ffv) node[midway, right] {$\textup{ad}f$};
    \draw[->] (ftildv) -- (ftildfv) node[midway, left] {$\textup{ad}f$};
    \draw[->] (v) -- (ftildv) node[midway, sloped, above] {$-\textup{ad}\tilde{f}$};
    \draw[->] (fv) -- (ftildfv) node[midway, above] {$\textup{ad}\tilde{f}$};
    \draw[->] (ftildv) -- (ffv);
    \draw[<->] (ftildv) -- (fv) node[midway, above] {$\textup{ad}U$};
  \end{tikzpicture}
  \qquad\qquad
    \begin{tikzpicture}
    \node (etv) at (-2,1) {$-[\tilde{e},a]$};
    \node (eetv) at (-2,3) {$-[e,[\tilde{e},a]]$};
    \node (Ev) at (1,3) {$[E,a]$};
    \node (ev) at (1,1) {$[e,a]$};
    \node (v) at (1,-1) {$a$};
    \draw[->] (v) -- (ev) node[midway, right] {$\textup{ad}e$};
    \draw[->] (ev) -- (Ev) node[midway, right] {$\textup{ad}e$};
    \draw[->] (etv) -- (eetv) node[midway, left] {$\textup{ad}e$};
    \draw[->] (v) -- (etv) node[midway, sloped, above] {$-\textup{ad}\tilde{e}$};
    \draw[->] (ev) -- (eetv) node[pos=.65, above] {$\quad\textup{ad}\tilde{e}$};
    \draw[->] (etv) -- (Ev);
    \draw[<->] (etv) -- (ev) node[midway, above] {$-\textup{ad}U$};
  \end{tikzpicture}
\end{center}
\end{enumerate}
\end{lemma}
\begin{proof}
Both statements can be directly deduced from Jacobi identity of Lie bracket.
\end{proof}

Since $\g(0)$ is abelian, any nonzero element in $\g(0)$ satisfies the assumption of Lemma \ref{lemma:sl repn}. Therefore, the highest weight vector $v_{2i}=v_{2i}^{(0)}$ of $R_{2i}$ in \eqref{eq:fix element in V_ij} is also in the kernel of $\textup{ad}U$ by Lemma \ref{lemma:sl repn} (1). Now we apply (2) of Lemma \ref{lemma:sl repn} repeatedly to get the following commutative diagram:
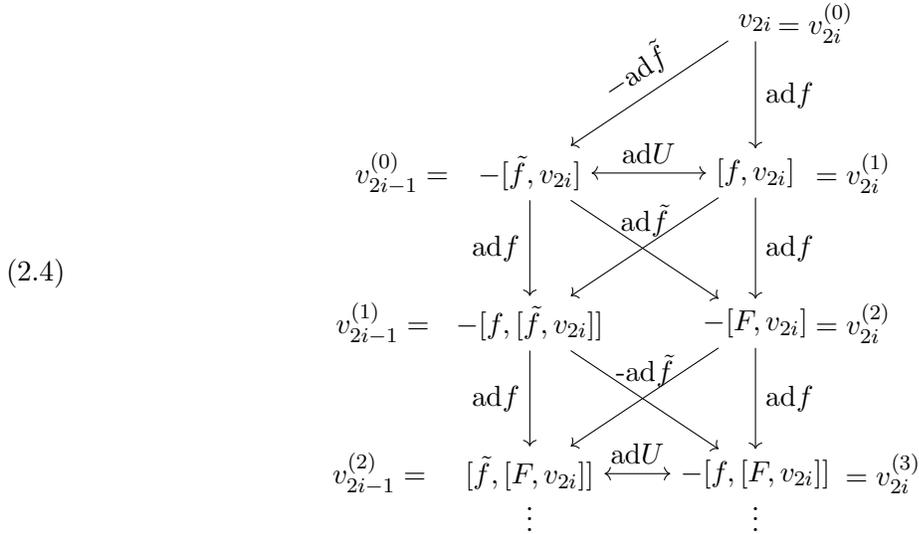
\begin{equation} \label{pic:sl repn}
\begin{aligned}
  \begin{tikzpicture}
    \node (ftildv) at (-1,1) {$-[\tilde{f},v_{2i}]$};
    \node  at (-2.7,1) {$v_{2i-1}^{(0)}=$};
    \node (ftildfv) at (-1,-1) {$-[f,[\tilde{f},v_{2i}]]$};
    \node  at (-2.7,-1) {$v_{2i-1}^{(1)}=$ \quad \quad };
    \node (v) at (2,3) {$v_{2i}$};
    \node  at (2.8,3) {$=v_{2i}^{(0)}$};
    \node (fv) at (2,1) {$[f,v_{2i}]$};
    \node  at (3.3,1) {$=v_{2i}^{(1)}$};
    \node (ffv) at (2,-1) {$-[F,v_{2i}]$};
    \node  at (3.3,-1) {$=v_{2i}^{(2)}$};
    \node (ftildFv) at (-1,-3) {$[\tilde{f},[F,v_{2i}]]$};
    \node  at (-3,-3) {$v_{2i-1}^{(2)}=$};
    \node (fFv) at (2,-3) {$-[f,[F,v_{2i}]]$};
    \node  at (3.7,-3) {$=v_{2i}^{(3)}$};
    \node (dots) at (-1,-3.5) {$\vdots$};
    \node (dots2) at (2,-3.5) {$\vdots$};
    \draw[->] (v) -- (fv) node[midway, right] {$\text{ad}f$};
    \draw[->] (fv) -- (ffv) node[midway, right] {$\text{ad}f$};
    \draw[->] (ftildv) -- (ftildfv) node[midway, left] {$\text{ad}f$};
    \draw[->] (ffv) -- (fFv) node[midway, right] {$\text{ad}f$};
    \draw[->] (ftildfv) -- (ftildFv) node[midway, left] {$\text{ad}f$};
    \draw[->] (v) -- (ftildv) node[midway, sloped, above] {$-\text{ad}\tilde{f}$};
    \draw[->] (fv) -- (ftildfv) node[midway, above] {$\text{ad}\tilde{f}$};
    \draw[->] (ftildv) -- (ffv);
    \draw[<->] (ftildv) -- (fv) node[midway, above] {$\text{ad}U$};
    \draw[<->] (ftildFv) -- (fFv) node[midway, above] {$\text{ad}U$};
    \draw[->] (ffv) -- (ftildFv) node[midway, above] {-$\text{ad}\tilde{f}$};
    \draw[->] (ftildfv) -- (fFv);
  \end{tikzpicture}
\end{aligned}
\end{equation}
There is a similar picture for the lowest weight vector $v_{2i}^{(2i)}$ and ad$e$, ad$\tilde{e}$ operations with upward arrows.

Note that the elements in the diagram with no adjacent horizontal arrows are killed by the adjoint action of $U$. In addition, $[\tilde{f}, v_{2i}]\in \ker(\text{ad}e)\cap \g(i-\frac{1}{2})= \CC v_{2i-1}$ is a nonzero element in $R_{2i-1}$.
Finally, we get the following irreducible decomposition of $\g$ as an $\mathfrak{sl}(2|1)$-representation.

\begin{theorem}
  Consider the principal $\mathfrak{sl}(2|1)$-embedding in $\g$. Then 
  \begin{equation}
    \g=\bigoplus_{i=1}^n \widetilde{R}_i,\quad \widetilde{R}_i=R_{2i-1}\oplus R_{2i}
  \end{equation}
   gives the decomposition of $\g$ into the direct sum of irreducible representations with respect to the adjoint action of the principal $\mathfrak{sl}(2|1)$. 
\end{theorem}
\begin{proof}
  Lemma \ref{lemma:sl repn} and the diagram \eqref{pic:sl repn} assure that each $\widetilde{R}_i$ is closed under the action of $\mathfrak{sl}(2|1)$ except for that of $\tilde{e}$. Hence it remains to show that $[\tilde{e}, v_{2i}]=0$, that we prove by a downward induction on $i$. 
  
  It is clear that $[\tilde{e}, v_{2n}]=0$. Considering the grading \eqref{sl(n+1|n) eigenspace decomposition} of $\g$, we have 
  \begin{equation*}
    [\tilde{e}, v_{2n-2i}]=\sum_{j=0}^{2i-1}k_j (\textup{ad}f)^j v_{2n-2i+1+j}
  \end{equation*}
  for some constants $k_j$'s. Apply $\text{ad}e$ on the both sides to get
  \begin{equation*}
    0=\sum_{j=1}^{2i-1}k_j (\textup{ad}e)(\textup{ad}f)^j v_{2n-2i+1+j}.
  \end{equation*}
Then, since $\{(\text{ad}e)(\text{ad}f)^j v_{2n-2i+1+j}|j=1,2,\cdots, 2i-1\}$ are linearly independent, all $k_j$ for $j\geq 1$ should be zero. In other words,
$ [\tilde{e},v_{2n-2i}]=k_0\, v_{2n-2i+1}.$
By the induction hypothesis, we have 
\begin{align*}
  0=k_0[\tilde{e},v_{2n-2i+1}]&=k_0 c_{n-i+1}^{-1}\big[\tilde{e},[\tilde{f},v_{2n-2i+2}]\big]=k_0 c_{n-i+1}^{-1}\big[[\tilde{e},\tilde{f}],v_{2n-2i+2}\big]\\
  &=-k_0 c_{n-i+1}^{-1}\big[2\fh ,v_{2n-2i+2}\big]=-(2n-2i+2)k_0 c_{n-i+1}^{-1}v_{2n-2i+2},
\end{align*}
where the nonzero constant $c_j$ is defined by
$[\tilde{f}, v_{2j}]=c_j v_{2j-1}$.
Thus, $k_0$ is also zero.
\end{proof}


\section{Poisson vertex algebras and Supersymmetric Poisson vertex algebras\label{sec: PVA and SUSY PVA}}

In this section, we review the definition of supersymmetric Poisson vertex algebras. 
For a detailed explanation or further relation between supersymmetric vertex algebras and supersymmetric Poisson vertex algebras, see the paper \cite{HK07} of Heluani and Kac.


\subsection{Basic notations} \label{sec: basic notations}
 In this section, we establish basic notations used throughout this paper. Let $P=P_{\bar{0}}\oplus P_{\bar{1}}$ be any vector superspace and $a\in P$ be a homogeneous element. If $a\in P_{\bar{0}}$ (resp. $a\in P_{\bar{1}}$), then the parity is defined by $p(a)=0$ (resp. $p(a)=1$).

 Let us consider endomorphisms on $P$, denoted as $D_i$ and $\partial$, where $D_i$ is odd and $\partial$ is even. The tuple of endomorphisms $D_i$ and $\partial$ is denoted by 
 \begin{equation} \label{eq: nabla}
  \nabla:=\left(\partial, D_1, \cdots D_N\right),
  \end{equation}
  provided that they satisfy the relations
  \begin{equation} \label{eq: nabla relation}
    [D_i, D_j]=2\delta_{i,j}\partial, \quad [\partial, D_i]=0.
  \end{equation}
  Additionally, we let $\CC[\nabla]$ be the unital superassociative algebra generated by $\nabla$. 
  
  We also consider unital superassociative algebra $\CC[\Lambda]$ generated by 
 \begin{equation} \label{eq:Lambda}
 \Lambda:=(\lambda, \chi_1, \cdots, \chi_N)
 \end{equation}
  subject to the relation
 \begin{equation} \label{eq: Lambda}
   [\lambda, \chi_i]=0, \quad [\chi_i, \chi_j]=-2\delta_{i,j}\lambda,
 \end{equation}
 where $\lambda$ is even and $\chi_i$'s are odd. One can view the vector superspace $\CC[\Lambda]\otimes P$ as a $\CC[\nabla]$-module with
 \begin{equation} \label{eq: D chi relation}
  D_i(\chi_j a)=-\chi_j D_i(a)+2\delta_{ij}\lambda a
 \end{equation}
 for $a\in P$, where the tensor product notation is omitted.
\subsection{Supersymmetric Poisson vertex algebras} \label{sec: N_K=N PVAs}

\begin{definition} \label{def: SUSY LCA}
  A \textit{$N_K=N$ supersymmetric (SUSY) Lie conformal algebra} is a $\ZZ / 2\ZZ$-graded $\CC[\nabla]$-module $R$ with a $\mathbb{C}$-linear map $\{\cdot \,{}_{\Lambda}\,\cdot\}: R \otimes R \rightarrow \mathbb{C}[\Lambda] \otimes_{\mathbb{C}} R$ of parity $N \Mod 2$ satisfying the following three axioms:
  \vskip 2mm
  \begin{enumerate}[]
   \item 
   (sesquilinearity) 
    $  \left\{D_i a{}_{\Lambda} b\right\} =-(-1)^N \chi_i\{a{}_{\Lambda} b\}, \{a_{\Lambda} D_i b\} =(-1)^{p(a)+N}\left(D_i+\chi_i\right)\{a_{\Lambda} b\},$
    \vskip 2mm
   \item (skewsymmetry)
       $\left\{b{}_{\Lambda} a\right\}=-(-1)^{p(a)p(b)+N}\left\{b{}_{-\Lambda-\nabla} a\right\},$
    \vskip 2mm
   \item (Jacobi identity) 
          \center{
        $ \left\{a_{\Lambda}\left\{b_{\Gamma} c\right\}\right\}=(-1)^{(p(a)+1)N}\left\{\left\{a_{\Lambda} b\right\}_{\Gamma+\Lambda} c\right\}+(-1)^{(p(a)+N)(p(b)+N)}\left\{b_{\Gamma}\left\{a_{\Lambda} c\right\}\right\}.$}
        \vskip 2mm
  \end{enumerate}
  for $a,b\in R$, $i=1,2,\cdots, N$ and $\Lambda,\Gamma$ two different tuples of the form \eqref{eq:Lambda}. 
  
  To be more precise, the axioms are understood in the following way. 
  \begin{enumerate}[(i)]
  \item The RHS of  the sesquilinearity and skewsymmetry relations can be computed by the $\CC[\nabla]$-module structure of $\CC[\Lambda]\otimes R$,
  \item 
The Jacobi identity is an equality in $\CC[\Lambda, \Gamma] \otimes R$ and an element in $\Lambda$ supercommutes with an element in $\Gamma$.
  \end{enumerate}
\end{definition}
  
  By performing the sesquilinearity twice, we get  
  \begin{equation} \label{eq:sesqui-PVA}
  \{\partial a{}_{\Lambda} b\}=-\lambda\{a{}_{\Lambda}b\}, \quad \{a{}_{\Lambda} \partial b\}=(\partial+\lambda)\{a{}_{\Lambda}b\},
  \end{equation}
 which are known as sesquilinearities of PVA. In particular, when $N=0$, an $\CC[\partial]$-module $R$ satisfying the skewsymmetry, Jacobi identity and \eqref{eq:sesqui-PVA} is simply called a Lie conformal algebra. 
In this case, we usually denote the $\Lambda$-bracket by $\{\cdot {}_{\lambda} \cdot\}$.

 If the superspace has in addition a unital supercommutative, superassociative $\CC$-algebra structure, which is compatible with the $\Lambda$-bracket, then we call it the $N_K=N$ supersymmetric Poisson vertex algebra.
\begin{definition} \label{def: N_K=N PVA}
  A \textit{$N_K=N$ supersymmetric (SUSY) Poisson vertex algebra (PVA)} is a tuple $\left(P,1, \nabla, \{\cdot {}_{\Lambda} \cdot\}, \cdot\right)$, where
  \begin{itemize}
    \item $\left(P, \nabla, \{\cdot {}_{\Lambda} \cdot\}\right)$  is a $N_K=N$ SUSY Lie conformal algebra,
    \item $\left(P,1, D_i, \cdot\right)$ is a unital supercommutative superassociative differential algebra with odd derivations $D_i$ for $i=1,2,\cdots, N$,
    \item (Leibniz rule) $\left\{a_{\Lambda} b c\right\}=\left\{a_{\Lambda} b\right\} c+(-1)^{(p(a)+N)p(b)} b\left\{a_{\Lambda} c\right\}$ for any $a,b,c\in P$.
  \end{itemize}
  In particular,  the $N_K=0$ SUSY PVA is simply called a PVA.
\end{definition}

\begin{definition} \label{def: superconformal vector}
  Let $P$ be a $N_K=N$  SUSY PVA. An element $G\in P$ of parity $N$ $(\text{mod}\ 2)$ is called a \textit{$N_K=N$  superconformal vector of central charge $c\in \CC$} if it satisfies the two following  conditions.
  \begin{enumerate}[(i)]
    \item The $\Lambda$-bracket between $G$ and itself is
    \begin{equation} \label{eq: conformal Lambda-bracket}
      \{G{}_{\Lambda}G\}=\left(2\partial+(4-N)\lambda+\sum_{i=1}^N \chi_i D_i\right)G+\frac{\lambda^{3-N}\chi_1 \chi_2\cdots \chi_N}{3}c,
    \end{equation} 
    where $c=0$ if $N\geq 4$.
    \item Denote by $P_\Delta$ for $\Delta\in \RR$ the subspace of $P$ consisting of elements $a\in P$ satisfying     
    \begin{equation} \label{eq: conformal weight Delta}
      \{G{}_{\Lambda}a\}=\left(2\partial+2\Delta \lambda +\sum_{i=1}^N\chi_i D_i\right)a+O(\Lambda^2).
    \end{equation}     
    Here $O(\Lambda^2)$ is a polynomial in $\Lambda$, whose constant and linear terms are vanishing. 
    Then 
    $P=\bigoplus_{\Delta\geq M} P_{\Delta}$ for some $M\in \RR$. 
  \end{enumerate}
If $P$ has a superconformal vector, $P=\bigoplus P_{\Delta}$ is called the \textit{conformal weight decomposition} of $P$. In addition, the value $\Delta$ in \eqref{eq: conformal weight Delta} is called the conformal weight of $a$ with respect to $G$. In the physics literature, an element $a$ satisfying \eqref{eq: conformal weight Delta} is called {\it quasi-primary}. Moreover, if $O(\Lambda^2)=0$ in \eqref{eq: conformal weight Delta},  we say 
that $a$ is {\it ($G-$)primary} of conformal weight $\Delta$.
 For later use, we note that $a \in P_\Delta$ is primary if and only if 
\begin{equation}\label{eq: primary conformal weight Delta-2}
  \{a{}_{\Lambda}G\} = (-1)^{p(a)N+N}\left( (2\Delta-2+N)\partial + 2\Delta \lambda + \sum_{i=1}^N \chi_i D_i \right) a.
\end{equation}
Additionally, when $N=0$, $G$ is commonly referred to as a \textit{conformal vector}. 

\begin{remark} \label{rem:conformal}
In many references (see e.g. \cite{Kac98}), when $N=0$, a conformal vector $L$ of a PVA is defined through the Virasoro relation:
\begin{equation} \label{eq:e-momentum}
 \{L{}_\lambda L\} =(\partial+ 2\lambda) L + \frac{c}{6}\lambda^3
\end{equation}
for some constant $c$. The RHS of \eqref{eq:e-momentum} is the half of our relation \eqref{eq: conformal Lambda-bracket} for the conformal vector. Then, an element $L$ in a PVA is a conformal vector with respect to the relation \eqref{eq:e-momentum} if and only if $2L$ is a conformal vector with respect to \eqref{eq: conformal Lambda-bracket}. However, the two definitions of conformal vectors  give rise to the same conformal weight decomposition.
\end{remark}

\end{definition}

\begin{remark}
  There is another kind of SUSY PVA called a $N_W$ SUSY PVA introduced in \cite{HK07}.
  However, in this paper, we only deal with $N_K=0,1,2$ SUSY PVAs and simply denote them by $N=0,1,2$ SUSY PVAs.
  \end{remark}

\subsection{Relation between SUSY PVAs} \label{sec: relation between PVAs}
To understand the structure of SUSY PVAs, we observe the relation between them. Especially, we see the relations between $N=0,1$ and $2$ Poisson structures. In the rest of this paper, we fix
\begin{equation} \label{eq: nabla lambda fixed notation}
  \begin{aligned}
    \bm{\nabla}=(\partial, D, \tilde{D}),\quad &\bm{\Lambda}=(\lambda, \chi, \tilde{\chi}),\\
    \nabla=(\partial, D),\quad \Lambda=(\lambda, \chi), \quad &\tilde{\nabla}=(\partial, \tilde{D}),\quad \tilde{\Lambda}=(\lambda, \tilde{\chi}),
  \end{aligned}
\end{equation}
all of which satisfy the relations introduced in section \ref{sec: basic notations}.
 We also adopt a uniform notation for the coefficients of each $N=0,1,2$ Poisson bracket as
\begin{equation} \label{eq: bracket fixed notation}
\begin{aligned}
  &\{a{}_{\lambda}b\}=\sum_{n\geq 0}\frac{\lambda^n}{n!}a_{(n)}b, \quad \{a{}_{\Lambda}b\}=\sum_{n\geq 0} \frac{\lambda^n}{n!}\left(a_{(n|0)}b+\chi a_{(n|1)}b\right),\\
  &\{a{}_{\bm{\Lambda}}b\}=\sum_{n\geq 0}\frac{\lambda^n}{n!}\left(a_{(n|00)}b-\chi a_{\left(n|10\right)}b+\tilde{\chi} a_{\left(n|01\right)}b-\chi \tilde{\chi}a_{\left(n|11\right)}b \right).
\end{aligned}
\end{equation}

\begin{pdefinition} \label{def: induced bracket}
Suppose $P$ is a $N=1$ or $2$ SUSY PVA. Then $P$ is a $N=0$ PVA 
    with the $\lambda$-bracket $\{\cdot {}_{\lambda}\cdot \}:P\otimes P \rightarrow \CC[\lambda] \otimes P$
    defined by 
    \begin{equation}\label{eq: induced lambda bracket from SUSY}
        \{a{}_{\lambda}b\}:=\sum_{n\geq 0}\frac{\lambda^n}{n!}a_{(n|\ri)}b, \quad \ri= \left\{\begin{array}{ll}
        1 &\text{ if } N=1,\\
        11 &\text{ if }N=2.  
        \end{array}\right.
    \end{equation}
    In the spirit of \eqref{eq: bracket fixed notation},  the $\lambda$-bracket in \eqref{eq: induced lambda bracket from SUSY} is equivalent to $a_{(n)}:=a_{(n|\ri)}$.\\
    Moreover, if $P$ is a $N=2$ SUSY PVA, each of the following $\Lambda$ or $\tilde{\Lambda}$-bracket
    \begin{equation} \label{eq: N=1 induced from N=2}
      \{a{}_{\Lambda}b\}=\sum_{n\geq 0} \frac{\lambda^n}{n!}\left(a_{(n|01)}b+\chi a_{(n|11)}b\right), \quad \{a{}_{\tilde{\Lambda}}b\}=\sum_{n\geq 0} \frac{\lambda^n}{n!}\left(a_{(n|10)}b+\tilde{\chi}a_{(n|11)}b\right)
    \end{equation}
    makes $P$ a $N=1$ SUSY PVA.
\end{pdefinition}
\begin{proof}
  Note that \eqref{eq: induced lambda bracket from SUSY} is obtained by selecting only those terms from the $\Lambda$ (resp. $\bm{\Lambda}$)-bracket in \eqref{eq: bracket fixed notation} that contain the variable $\chi$ (resp. $\chi\tilde{\chi}$). Therefore, one can show that \eqref{eq: induced lambda bracket from SUSY} is indeed a Poisson bracket by comparing the coefficients of $\chi$ (resp. $\chi\tilde{\chi}$) in each of \eqref{eq:sesqui-PVA} and conditions in Definition \ref{def: SUSY LCA}. The last statement is also true for a similar reason.
\end{proof}
In the rest of this paper, even in the context of $N=1$ or $2$ SUSY PVA, we will use the expression $a_{(n)}$ to refer to $a_{(n|\ri)}$ for $\ri$ in \eqref{eq: induced lambda bracket from SUSY}. In addition, focusing on the interrelation between $N=0,1,$ and $2$ SUSY structures, one can check by sesquilinearity that \eqref{eq: bracket fixed notation}, \eqref{eq: induced lambda bracket from SUSY} and \eqref{eq: N=1 induced from N=2} are correlated as
\begin{align}
  \{a{}_{\Lambda}b\}=&\{Da{}_{\lambda}b\}+\chi\{a{}_{\lambda}b\},\quad \{a{}_{\tilde{\Lambda}}b\}=\{\tilde{D}a{}_{\lambda}b\}+\tilde{\chi}\{a{}_{\lambda}b\}, \label{eq: N=1 induced from nonsusy}\\
  \{a{}_{\bm{\Lambda}}b\}=&\{ \tilde{D} D a{}_{\lambda}b\}-\chi\{\tilde{D} a{}_{\lambda}b\}+\tilde{\chi}\{D a{}_{\lambda}b\}-\chi \tilde{\chi} \{a{}_{\lambda}b\} \label{eq: N=2 induced from nonsusy}\\
      =&-\{\tilde{D}a{}_{\Lambda}b\}+\tilde{\chi} \{a{}_{\Lambda}b\}=\{Da{}_{\tilde{\Lambda}}b\}-\chi\{a{}_{\tilde{\Lambda}}b\}, \nonumber
\end{align}
when $P$ is a $N=1$ or $2$ SUSY PVA. For example, when $P$ is a $N=1$ SUSY PVA, one can obtain $(Da)_{(n|1)}=a_{(n|0)}$ by comparing the coefficients of $\chi$ in $\{Da{}_{\Lambda}b\}=\chi\{a{}_{\Lambda}b\}$. Then this result demonstrates the first equality in \eqref{eq: N=1 induced from nonsusy}.

Now, it is natural to ask whether the existence of an additional odd derivation on $N=0$ or $1$ SUSY PVA would lead to an additional SUSY structure if we use \eqref{eq: N=1 induced from nonsusy} or \eqref{eq: N=2 induced from nonsusy}. To investigate this possibility, we define the induced brackets in an obvious way, when a $N=0$ or $1$ SUSY PVA has an odd derivation $D$ or $\tilde{D}$. Each bracket induced by $\Lambda$, $\tilde{\Lambda}$, $\bm{\Lambda}$ will be denoted by
\begin{equation*}
  \{\cdot{}_{\Lambda}\cdot \}^-, \quad \{\cdot{}_{\tilde{\Lambda}}\cdot\}^-,\quad \{\cdot{}_{\bm{\Lambda}}\cdot\}^-,
\end{equation*}
respectively. The induced bracket $\{\cdot{}_{\Lambda}\cdot\}^-$ can be defined on any PVA when it is equipped with any odd derivation $D$. However, these induced brackets may not be Poisson brackets. If the induced brackets are indeed Poisson brackets, then they satisfy the relations \eqref{eq: N=1 induced from nonsusy} and \eqref{eq: N=2 induced from nonsusy} again. The following propositions provide criteria for induced brackets  to be Poisson brackets.

\begin{proposition} \label{prop: pva, N=1 pva relation}
  Let $P$ be a PVA. The following two statements are equivalent:
  \begin{enumerate}
    \item  $(P, 1, \nabla, \{\cdot {}_{\Lambda}\cdot\}^-, \cdot)$ is a $N=1$ SUSY PVA,
    \item  There exists an odd derivation $D:P\rightarrow P$ satisfying 
    \begin{equation} \label{eq: supersymmetry condition}
      D^2=\partial, \quad D\{a{}_{\lambda}b\}=\{Da{}_{\lambda}b\}+(-1)^{p(a)}\{a{}_{\lambda}Db\}
    \end{equation}    
    for any $a,b\in P$.
  \end{enumerate}
  Here, we assume that $\{\cdot {}_{\Lambda}\cdot\}^-$ is induced by $\{\cdot {}_{\lambda}\cdot\}$ as in \eqref{eq: N=1 induced from nonsusy}.
\end{proposition}
\begin{proof}
  Let $P$ be a $N=1$ SUSY PVA. Comparing the coefficients of $\chi$ in the sesquilinearity $D\{a{}_{\Lambda}b\}^--(-1)^{p(a)+1}\{a{}_{\Lambda}Db\}^-=-\chi\{a{}_{\Lambda}b\}^-$, one gets
  \begin{equation} \label{eq: [D, a_(n|1)]}
    Da_{(n|1)}b+(-1)^{p(a)+1}a_{(n|1)}Db=a_{(n|0)}b.
  \end{equation}
  In terms of the non-SUSY structure of $P$, \eqref{eq: [D, a_(n|1)]} is equivalent to \eqref{eq: supersymmetry condition}.
  Now, assume (2) and check if the induced $\Lambda$-bracket $\{\cdot{}_{\Lambda}\cdot\}^-$ satisfies sesquilinearity, skewsymmetry and Jacobi identity. By computation, we have 
  \begin{equation} \label{eq:N=0,1 lambda bracket, sesqui-1}
    \begin{aligned}
      \{ Da{}_\Lambda b \}^-& =\{\partial a{}_\lambda b\}+\chi \{ Da{}_\lambda b\} = \chi(\chi \{ a{}_\lambda b\} +\{ Da{}_\lambda b\})= \chi \{a{}_\Lambda b\}^-
    \end{aligned}
  \end{equation}
 and 
  \begin{equation} \label{eq: lambda bracket, D compatibility}
    \begin{aligned}
      \{a{}_{\lambda} Db\}&=\sum_{n\geq 0} \frac{\lambda^n}{n!} a_{(n)}(Db)=\sum_{n\geq 0} \frac{\lambda^n}{n!}\left([a_{(n)}, D]+(-1)^{p(a)} D\,  a_{(n)}\right)b\\
      &=\sum_{n\geq 0} \frac{\lambda^n}{n!}\left((-1)^{p(a)+1} (Da)_{(n)}b + (-1)^{p(a)} D(a_{(n)}b)\right)\\
      &=(-1)^{p(a)+1}\{Da{}_{\lambda}b\}+(-1)^{p(a)} D\{a{}_{\lambda}b\}.
    \end{aligned}
  \end{equation}
  By \eqref{eq: lambda bracket, D compatibility}, one can deduce
  \begin{equation}\label{eq:N=0,1 lambda bracket, sesqui-2}
    \begin{aligned}
    & \{ a{}_\Lambda Db \}^-= \chi\big((-1)^{p(a)+1}\{Da{}_\lambda b\}+(-1)^{p(a)}D\{a{}_\lambda b\}\big)\\
    & \hskip 10mm + \big((-1)^{p(a)}\{\partial a{}_\lambda b\} + (-1)^{p(a)+1}D\{Da{}_\lambda b\}\big)\\
    & \hskip 10mm = (-1)^{p(a)+1}(\chi+D) \big( \{Da{}_\lambda b\}+ \chi \{ a{}_\lambda b\} \big)= (-1)^{p(a)+1}(\chi+D) \{ a{}_\Lambda b \}^-.
    \end{aligned}
  \end{equation}
  Hence we get the sesquilinearity from \eqref{eq:N=0,1 lambda bracket, sesqui-1} and \eqref{eq:N=0,1 lambda bracket, sesqui-2}. The skewsymmetry of the $\Lambda$-bracket follows from the skewsymmetry of PVA $P$. More precisely, we have 
  \begin{equation*}
    \begin{aligned}
     \{a{}_{-\Lambda-\nabla}b\}^- & =\{Da{}_{-\lambda-\partial} b\} -(\chi+D) \{a{}_{-\lambda-\partial} b\}\\
     & = (-1)^{p(b)(p(a)+1)+1}\{b{}_\lambda Da\}-(\chi+D)(-1)^{p(a)p(b)+1}\{b{}_\lambda a\}\\
     & = (-1)^{p(a)p(b)} \big( \{Db{}_\lambda a\} - D\{b{}_\lambda a\}+(\chi+D)\{b{}_\lambda a\}\big) \\
     & = (-1)^{p(a)p(b)}\{b{}_\Lambda a\}^-.
    \end{aligned}
  \end{equation*}
  Similarly, the Jacobi identity of $\Lambda$-bracket follows from that of $\lambda$-bracket and \eqref{eq: lambda bracket, D compatibility}.

\end{proof}

\begin{proposition} \label{prop: pva, N=1, N=2 pva relation}
  Let $P$ be a PVA. The following three statements are equivalent.
  \begin{enumerate}
    \item $(P, 1, \bm{\nabla}, \{\cdot {}_{\bm{\Lambda}} \cdot\}^-, \cdot)$ is a $N=2$ SUSY PVA,
    \item  $(P, 1, \nabla, \{\cdot {}_{\Lambda} \cdot \}^-, \cdot)$ is a $N=1$ SUSY PVA with an odd derivation $\tilde{D}$ on $P$ satisfying $\tilde{D}^2=\partial$, $[D, \tilde{D}]=0$ and $\tilde{D}\{a{}_{\Lambda}b\}^-=-\{\tilde{D}a{}_{\Lambda}b\}^-+(-1)^{p(a)+1}\{a{}_{\Lambda}\tilde{D}b\}^-$ for any $a,b\in P$,
    \item  $P$ has two supercommuting odd derivations $D$ and  $\tilde{D}$ satisfying $D^2=\tilde{D}^2=\partial$ and $D\{a{}_{\lambda}b\}=\{Da{}_{\lambda}b\}+(-1)^{p(a)}\{a{}_{\lambda}Db\}$, $\tilde{D}\{a{}_{\lambda}b\}=\{\tilde{D}a{}_{\lambda}b\}+(-1)^{p(a)}\{a{}_{\lambda}\tilde{D}b\}$ for any $a,b\in P$.
  \end{enumerate}
\end{proposition}
\begin{proof}
  The equivalence of (2) and (3) is immediate from Proposition \ref{prop: pva, N=1 pva relation}, so we only need to show that (1) and (2) are equivalent. If $P$ is a $N=2$ SUSY PVA, then the $N=1$ SUSY PVA structure of $P$ is guaranteed by Proposition-Definition \ref{def: induced bracket}. Additionally, the properties of $\tilde{D}$ can be checked using sesquilinearity, as illustrated in the proof of Proposition \ref{prop: pva, N=1 pva relation}. One can prove the opposite by the computation similar to \eqref{eq:N=0,1 lambda bracket, sesqui-2}.
\end{proof}
In the following, when induced brackets are Poisson brackets, we shall omit the ``$-$" symbol in the induced brackets.

Thus, to find an additional SUSY structure, one has first to identify an odd derivation $D$ that satisfies \eqref{eq: supersymmetry condition}. For any given element $\tau$ in a PVA, $\tau_{(0)}$ is a derivation that meets the second condition of \eqref{eq: supersymmetry condition}, as it can be deduced from Jacobi identity and Leibniz rule. Therefore, we can take $D:=\tau_{(0)}$ whenever $\tau$ is an odd element satisfying $[\tau_{(0)}, \tau_{(0)}]=(\tau_{(0)}\tau)_{(0)}=2\partial$.

At this point, conformal vectors come into play. Note that the last condition for $\tau$ is achieved when $\tau_{(0)}\tau$ is a conformal vector. Hence, we assume that
\begin{enumerate}
\item $\tau_{(0)}\tau$ is indeed a conformal vector, allowing us to put $D:=\tau_{(0)}$,
\item $\tau$ is primary of conformal weight $\frac{3}{2}$ with respect to $\tau_{(0)}\tau$. 
\end{enumerate}
Then, in terms of the induced $\Lambda$-bracket, we get an equation
\begin{equation} \label{eq: remind us of superconformal}
  \{\tau{}_{\Lambda}\tau\}^-=(2\partial+3 \lambda+\chi D)\tau +O(\Lambda^2),
\end{equation}
which has to be compared with the Definition \ref{def: superconformal vector} for a $N=1$ superconformal vector. The following theorem shows how to get a superconformal vector along with an additional SUSY structure and determines the conformal weight decomposition of $P$.

\begin{theorem} \label{thm: pva, N=1 pva relation with superconformal}
Let $P$ be a PVA with a conformal vector $L$ of central charge $c\in \CC$.  Suppose that there are primary elements $G$ and $a_i,\,i\in {\cI}$, such that $G$ is odd,
\begin{equation} \label{eq: nonsusy condition for N=1 in Thm}
  \{G{}_{\lambda} G\}=L+\frac{\lambda^2}{3}c, \quad \{G{}_{\lambda}a_i\}=G_{(0)}a_i,\quad  i\in {\cI},
\end{equation}
and the set $\{G,L\}\cup \{a_i, G_{(0)}a_i\}_{i\in {\cI}}$ freely generates $P$ as a $\CC[\partial]$-algebra. Then, for $D:=G_{(0)}$, $P$ is a $N=1$ SUSY PVA with a $N=1$ superconformal vector $G$. Moreover, $P$ is freely generated by the set $\{G\}\cup\{a_i\}_{i\in {\cI}}$ as a $\CC[\nabla]$-algebra, where $\{a_i\}_{i\in {\cI}}$ consists of $G$-primary elements.
\end{theorem}
\begin{proof}
  Given that the conformal weight of $G_{(0)}G=L$ is $2$, $G$ has conformal weight $\frac{3}{2}$. Therefore, by the preceding argument of this theorem, one can take $D:=G_{(0)}$ to get a $N=1$ SUSY structure. From the $L$-primarity of $G$ and the first condition in \eqref{eq: nonsusy condition for N=1 in Thm}, it follows that
  \begin{equation} \label{eq: nonsusy condition for N=1 superconformal}
    \{DG{}_{\lambda}G\}=(2\partial+3\lambda)G, \quad \{G{}_{\lambda}G\}=DG+\frac{\lambda^2}{3}c.
  \end{equation}
  This is equivalent to \eqref{eq: conformal Lambda-bracket} if we use \eqref{eq: N=1 induced from nonsusy}. The remaining conditions prove the $G$-primarity of $a_i$ in a similar way.
\end{proof}

Following a similar argument as in the proof of Theorem \ref{thm: pva, N=1 pva relation with superconformal}, one can also prove the following theorem.
\begin{theorem} \label{thm: N=1, N=2 pva relation with superconformal}
 Let $P$ be a $N=1$ SUSY PVA with a $N=1$ superconformal vector $G$ of central charge $c\in \CC$ . Suppose that there are primary elements $J$ and $a_i$, $i\in {\cI}$, such that $J$ is even,
 \begin{equation} \label{eq: N=1 condition for N=2 in Thm}
  \{J{}_{\Lambda}J\}=-G-\frac{\lambda\chi}{3}c,\quad \{J{}_{\Lambda}a_i\}=J_{(0|0)}a_i,
 \end{equation}
 and the set $\{J, G\}\cup \{a_i, J_{(0|0)}a_i\}$ freely generates $P$ as a $\CC[\nabla]$-algebra. Then, for $\tilde{D}:=J_{(0|0)}$, $P$ is a $N=2$ SUSY PVA with a $N=2$ superconformal vector $J$. Moreover, $P$ is freely generated by the set $\{J\}\cup \{a_i\}_{i\in {\cI}}$ as a $\CC[\bm{\nabla}]$-algebra, where $\{a_i\}_{i\in {\cI}}$ consists of $J$-primary elements.
\end{theorem}
\begin{proof}
  The proof is similar to that of Theorem \ref{thm: pva, N=1 pva relation with superconformal}. Note that the condition \eqref{eq: conformal Lambda-bracket} for $J$ to be a $N=2$ superconformal vector is equivalent to
  \begin{equation}
    \{\tilde{D}J{}_{\Lambda}J\}=-(2\partial+2\lambda+\chi D)J,\quad \{J{}_{\Lambda}J\}=\tilde{D}J-\frac{\lambda \chi}{3}c.
  \end{equation}
\end{proof}

\begin{remark} \ 
  \begin{enumerate} 
    \item   One can also show that each of the converse statements of Theorem \ref{thm: pva, N=1 pva relation with superconformal} and \ref{thm: N=1, N=2 pva relation with superconformal} is also true. For example, if $P$ is a $N=1$ SUSY PVA with a $N=1$ superconformal vector $G$, and $P$ is freely generated by $G$-primary elements and $G$ itself, then (i) $D=G_{(0)}$, (ii) $L:=DG$ is a conformal vector, (iii) \eqref{eq: nonsusy condition for N=1 in Thm} holds and (iv) $P$ is freely generated by $L$-primary elements and $L$ itself as $\CC[\partial]$-algebra. This converse statement can be shown using \eqref{eq: nonsusy condition for N=1 superconformal} and \eqref{eq: lambda bracket, D compatibility}.
    \item 
    Taking into account the converse of Theorem \ref{thm: N=1, N=2 pva relation with superconformal}, $-\tilde{D}J$ is a $N=1$ superconformal vector with the induced $\Lambda$-bracket, if $J$ is a $N=2$ superconformal vector. Changing the role of $(D, \chi)$ and $(\tilde{D}, \tilde{\chi})$, one can also say that $DJ$ is a $N=1$ superconformal vector, but with the $\tilde{\Lambda}$-bracket introduced in \eqref{eq: N=1 induced from N=2}.
  \end{enumerate}
\end{remark}

For later use, we emphasize here the equalities to check for an existence of a $N=2$ superconformal vector in terms of the PVA language.
\begin{theorem} \label{thm:N=2 superconformal in N=0 PVA}
  Let $P$ be a $N=2$ SUSY PVA. An even element $J$ of $P$ satisfies \eqref{eq: conformal Lambda-bracket} if and only if
  \begin{gather*}
    \{D\tilde{D}J{}_{\lambda}J\}=-(2\partial+2\lambda)J, \quad \{\tilde{D}J{}_{\lambda}J\}=-DJ, \quad \{DJ{}_{\lambda}J\}=\tilde{D}J, \quad \{J{}_{\lambda}J\}=-\frac{\lambda}{3}c.
  \end{gather*}
  Furthermore, for a $N=2$ superconformal vector $J$ of $P$, an element $a\in P$ is $J$-primary if and only if 
  \begin{gather*}
    \{D\tilde{D}J{}_{\lambda}a\}=-(2\partial+2\Delta\lambda)a, \quad \{\tilde{D}J{}_{\lambda}a\}=-Da, \quad \{DJ{}_{\lambda}a\}=\tilde{D}a, \quad \{J{}_{\lambda}a\}=0.
  \end{gather*}
\end{theorem}

\subsection{Reconstruction Theorem and Examples} \label{sec: Reconstruction Theorem and Examples}
In this section, we explain a way to get a freely generated SUSY PVA from $N_K=N$ a SUSY Lie conformal algebra and see some examples. We fix $\bm{\nabla}, \bm{\Lambda}, \nabla$ and $\Lambda$ as in \eqref{eq: nabla lambda fixed notation}. First, we have the well-known reconstruction theorem for $N=0$.
\begin{theorem}[Reconstruction Theorem \cite{Kac98}] \label{thm: S(R) pva}
  Let $R$ be a Lie conformal algebra. Then the supersymmetric algebra $S(R):=S(R_{\bar{0}})\otimes \bigwedge(R_{\bar{1}})$ is a PVA, whose $\lambda$-bracket is induced from the bracket of $R$ by Leibniz rule.
\end{theorem}
Using Proposition \ref{prop: pva, N=1 pva relation} and \ref{prop: pva, N=1, N=2 pva relation}, we get the supersymmetric version of the reconstruction theorem.
\begin{theorem}
  Let $R$ be a $N=1$ (resp. $N=2$) SUSY Lie conformal algebra. Then the supersymmetric algebra $S(R)$ has a $N=1$ (resp. $N=2$) SUSY PVA structure, whose $\Lambda$ (resp. $\bm{\Lambda}$)-bracket is induced from that of $R$ by Leibniz rule, and $D$ (resp. $D$ and $\tilde{D}$) is extended to the derivation on $S(R)$.
\end{theorem}
\begin{proof}
  For $N=1$ case, we use Proposition \ref{prop: pva, N=1 pva relation} and Theorem \ref{thm: S(R) pva}. Since $R$ is a Lie conformal algebra via the $\lambda$-bracket \eqref{eq: induced lambda bracket from SUSY}, the supersymmetric algebra $S(R)$ is a PVA by Theorem \ref{thm: S(R) pva}. Since the odd derivation $D$ on $S(R)$ 
  satisfies $[D, a_{(n)}]=(Da)_{(n)}$,  Proposition \ref{prop: pva, N=1 pva relation} tells that $S(R)$ is a $N=1$ SUSY PVA. For $N=2$ case, we use Proposition \ref{prop: pva, N=1, N=2 pva relation} and Theorem \ref{thm: S(R) pva}.
\end{proof}
Furthermore, the following proposition saves us some effort to check axioms of SUSY Lie conformal algebra.
\begin{proposition}[\cites{BDK09,CS,HK07}]
  Let $V$ be a vector superspace and $R= \CC[\partial]\otimes V$. 
  Suppose $V$ is 
  endowed with a linear map $\{\cdot {}_{\lambda}\cdot \}: V\otimes V \rightarrow \CC[\lambda]\otimes R$ of even parity. 
  \begin{enumerate}
  \item If $\{\cdot {}_{\lambda} \cdot\}$ on $V \otimes V$ satisfies skewsymmetry and Jacobi identity, then the bracket $\{\cdot {}_{\lambda}\cdot\}: R\otimes R \rightarrow \CC[\lambda]\otimes R$ extended from the bracket of $V$ by sesquilinearity gives $R$ a Lie conformal algebra structure. 
  \item Analogous statements are also true when we replace $(\lambda, \partial)$ by $(\Lambda, \nabla)$ or $(\bm{\Lambda}, \bm{\nabla})$, even parity map $\{\cdot {}_{\lambda} \cdot\}$ by an odd parity map $\{\cdot {}_{\Lambda} \cdot\}$ or an even parity map $\{\cdot {}_{\bm{\Lambda}} \cdot\}$, and Lie conformal algebra by a $N=1$ or $N=2$ SUSY Lie conformal algebra.
  \end{enumerate}
\end{proposition}

Now we end this section by introducing affine PVAs and $N=1$ SUSY affine PVAs and their (super)conformal vectors, which are the main ingredients of classical $W$-algebras and SUSY classical $W$-algebras.

\begin{example} [Affine PVA] \label{ex:affine}
  Let $\g$ be a finite simple Lie superalgebra endowed with an even supersymmetric invariant  bilinear form $(\cdot \,|\,\cdot)$.
 Suppose  $\CC[\partial]$-module $R(\g)=\CC[\partial]\otimes \g \oplus \CC K$ is a current Lie conformal algebra, i.e.  the $\lambda$-bracket $\{\cdot {}_{\lambda}\cdot \}: R(\g)\otimes R(\g) \rightarrow \CC[\lambda]\otimes R(\g)$ is defined by  
    \begin{equation*}
      \{a{}_{\lambda}b\}=[a,b]+(a|b)\lambda K, \quad \{K{}_\lambda a\}=\{K{}_\lambda a\}=0
    \end{equation*}
    for any $a,b\in \g$. The {\it affine Poisson vertex algebra of level $k\in \CC$ associated with $\g$} is a PVA defined by
    \begin{equation}
      V^k(\g)=S(R(\g))/\langle K-k \rangle,
    \end{equation}
    where $\langle K-k \rangle$ denotes the differential algebra ideal of $S(R)$ generated by $K-k$. Assuming $k\neq 0$, the affine PVA of level $k$ has a conformal vector
    \begin{equation} \label{eq:sugawara}
      L=\frac{1}{k}\sum_{i\in {\cI}}v^i v_i,
    \end{equation}
    where $\{v^i\}_{i\in {\cI}}$ and $\{v_i\}_{i\in {\cI}}$ are dual bases of $\g$ with respect to $(\cdot|\cdot)$ given by $(v^i|v_j)=\delta_{ij}$.
    The half of \eqref{eq:sugawara} is known as a \textit{Sugawara construction} (see Remark \ref{rem:conformal}).
    With respect to the conformal vector $L$, every element $a\in \g$ is primary of conformal weight $1$.
\end{example}
\begin{example}[$N=1$ SUSY Affine PVA] \label{ex: susy affine pva}
 Let $\g$ be the Lie superalgebra in Example \ref{ex:affine} and let $R(\bar{\g})=\CC[\nabla]\otimes \bar{\g} \oplus \CC K$ be the SUSY current Lie conformal algebra, where $\bar{\g}$ denotes $\g$ with reversed parity; that is the $\Lambda$-bracket on $\bar{\g}$ is  defined by 
    \begin{equation*}
      \{\bar{a}{}_{\Lambda}\bar{b}\}=(-1)^{p(a)}\left(\overline{[a,b]}+(a|b)\chi K\right)
    \end{equation*}
    for $a,b\in \g$ and $K$ is central. The $N=1$ affine PVA of level $k\in \CC$ associated with $\g$ is
    \begin{equation}
      V^k(\bar{\g}):=S(R(\bar{\g}))/\langle K-k \rangle,
    \end{equation} 
    where $\langle K-k \rangle$ is the differential algebra ideal of $S(R(\bar{\g}))$. If $k\neq 0$, it also has a $N=1$ superconformal vector
    \begin{equation} \label{eq:Kac-Todorov}
      \tau=\frac{1}{k}\left(\sum_{i\in {\cI}} (-1)^{i}D(\overline{u}^i)\overline{u}_i+\sum_{i,j,t\in {\cI}}\frac{1}{3k}(-1)^{p(i)p(t)}(u^i|[u^j,u^t])\overline{u}_i\overline{u}_j\overline{u}_t\right),
    \end{equation} 
    where $\{u^i\}_{i\in {\cI}}$ and $\{u_i\}_{i\in {\cI}}$ are dual bases of $\g$ with respect to $(\cdot|\cdot)$ given by $(u^i|u_j)=\delta_{ij}$ and $p(i):= p(u^i)$ for $i\in {\cI}$. This is called a \textit{Kac-Todorov construction}. In terms of $\tau$, every element $\bar{a}\in \bar{\g}$ is primary of conformal weight $\frac{1}{2}$.
\end{example}

\section{Poisson structures on (SUSY) \texorpdfstring{$W$}-algebras\label{sec:W-algebra}}

Let $\g$ be a simple finite dimensional Lie superalgebra with a nondegenerate even supersymmetric invariant bilinear form $( \cdot | \cdot )$. Suppose $\g$ has a subalgebra $\mathfrak{s}$ which is isomorphic to $\osp(1|2)$. We fix a $\CC$-basis $\{E,F,\fh ,e,f\}$ of $\mathfrak{s}$ satisfying the properties (osp-1)--(osp-4) in section \ref{sec: osp(1|2)-repn}, and assume that the bilinear form is normalized as $(E|F)=1$. Consider the $\ad \fh $-grading on $\g$:
\begin{equation} \label{grading}
\g= \bigoplus_{i\in \frac{1}{2}\ZZ} \g(i) 
\end{equation}
and the subspaces
\begin{equation} \label{parabol}
 \n = \bigoplus_{i>0}\g(i), \quad \pp = \bigoplus_{i\leq 0}\g(i), \quad \mathfrak{q}=\bigoplus_{i\leq \frac{1}{2}} \g(i).\end{equation}
Note that whenever we take an odd nilpotent element $f$ in a subalgebra isomorphic to $\osp(1|2)$,
normalized as in (osp-1)--(osp-4), it completely determines the grading \eqref{grading}.

\subsection{$N=1$ SUSY $W$-algebra $\WW^k(\bar{\g},f)$} \label{subsec:SUSY W-generator}

Recall $\Lambda$ and $\nabla$ in \eqref{eq: nabla lambda fixed notation}, and the $N=1$ SUSY affine PVA $V^k(\bar{\g})$ in Example \ref{ex: susy affine pva}.
Define the differential algebra homomorphism 
\begin{equation} \label{barho}
 \bar{\rho}: \CC[\Lambda]\otimes V^k(\bar{\g}) \to \CC[\Lambda]\otimes  S(\CC[\nabla]\otimes \bar{\pp})
\end{equation}
such that $\bar{\rho}(\lambda)=\lambda$, $\bar{\rho}(\chi)=\chi$, $\bar{\rho}(\bar{n})=(f|\bar n)$ and $\bar{\rho}(\bar{p})=\bar{p}$ for $\bar n\in \bar\n$ and $\bar{p}\in \bar{\pp}$. 

The $N=1$ SUSY classical $W$-algebra associated with $\g$ and $f$ is 
\begin{equation} \label{def:SUSY W}
 \mathcal{W}^k(\bar{\g},f)= \{ w\in S(\CC[\nabla]\otimes \bar{\pp}) \, | \, \bar{\rho}(\ad \bar{n} (w))=0 \text{ for } \bar n\in \bar\n\},
\end{equation}
where $\ad \bar{n} (w) = \{\bar{n}{}_\Lambda w\}$ for the $\Lambda$-bracket of $V^k(\bar{\g})$. 
One can check that $\mathcal{W}^k(\bar{\g},f)$ has a $N=1$ SUSY PVA structure with respect to the bracket induced from that of $V^k(\bar{\g})$. We will loosely write,
\[ \{w_1 \, {}_\Lambda \, w_2\} := \bar{\rho}(\{w_1 {}_\Lambda \, w_2\})\]
for $w_1, w_2 \in \mathcal{W}^k(\bar{\g},f)$, where the $\Lambda$-bracket on the RHS is the bracket of $V^k(\bar{\g})$.

Now, we describe $\frac{1}{2}\ZZ_+$-grading $\Delta$ on $S(\CC[\nabla]\otimes \bar{\pp})$ which also induces a grading on $\WW^k(\bar{\g},f)$, called a {\it conformal weight} decomposition. The $\Delta$-grading on $S(\CC[\nabla]\otimes \bar{\pp})$ is  defined by 
\begin{equation} \label{eq:conformal weight-SUSY W}
  \Delta(\bar{a})=\frac{1}{2}-j_a, \quad \Delta(AB)= \Delta(A)+ \Delta(B), \quad \Delta(D(A))= \frac{1}{2}+\Delta(A)
\end{equation}
for $a\in \g(j_a)$ and $A,B \in S(\CC[\nabla]\otimes \bar{\pp})$.
Note that one can consider the $\Delta$-grading of $\bar{a}$ for any $a\in\g^f$, since $\g^f$ is a subspace of $\pp$. This follows from the fact that $\g^f\subset \g^F$, and that \eqref{grading} is a good grading, which implies that $\g^F:= \ker (\ad F) \subset \pp$. For the defintion and the classification of good grading on basic Lie superalgebras, we refer to \cite{Hoyt12}.

\begin{remark}  
  For $k\neq 0$, the conformal weight decomposition of $\WW^k(\bar{\g},f)$ is indeed the decomposition obtained as in (2) of Definition \ref{def: superconformal vector}. Recall the superconformal vector $\tau$ of $V^k(\bar{\g})$ in \eqref{eq:Kac-Todorov}. The element $G_{\g,f}= \bar{\rho}(\tau)-2\partial \fh $ in $\WW^k(\bar{\g},f)$ is then a $N=1$ superconformal vector, whose conformal weight decomposition coincides with \eqref{eq:conformal weight-SUSY W}. However, at that point, it is still unknown whether the $W$-algebra is generated by $G_{\g, f}$ and $G_{\g,f}$-primary elements.
\end{remark}

\begin{proposition}  \label{prop:W-generator} \cite{Suh20}
Fix a basis $\{\tilde{v}^i_{(0)}|i\in {\cI}\}$ of $\g^f$ whose elements are homogeneous with respect to the grading \eqref{grading} and the parity. 
Then for each $i\in {\cI}$, there is an element $\w_i\in \WW^k(\bar{\g},f)$ whose
 linear part without total derivatives is $\bar{\tilde{v}}^i_{(0)}$ and the conformal weight is $\Delta(\w_i)= \Delta(\bar{\tilde{v}}^i_{(0)})$.
Moreover, such $\w_i$'s for $i\in {\cI}$ generate $\WW^k(\bar{\g},f)$ as a differential algebra.
\end{proposition}

The generators $\w_i$'s in Proposition \ref{prop:W-generator} can be determined using the method of Lax operator \cite{CS} or row determinant formula \cite{MRS21,RSS23} when $\g$ is of basic type and $f$ is principal. Note that the generators of the $W$-algebra satisfying Proposition \ref{prop:W-generator} may not be unique. For example, if $\Delta(\w_i)=\Delta(\w_j)-\frac{1}{2}$ for some $i,j\in {\cI}$, then 
$\widetilde{\w}_j= \w_j  + D(\w_i) \in \WW^k(\bar{\g},f)$ can replace $\w_j$.  

However, we can fix uniquely the elements $\w_i$ by imposing the following property. By the $\osp(1|2)$-representation theory, we have 
\begin{equation} \label{eq: decomp g^f}
  \g^f \oplus[e,\g]= \g.
\end{equation}
Hence any element in $S(\CC[\nabla]\otimes \bar{\g})$ can be uniquely decomposed in the following way:
\begin{equation}
  A= \sum_{m\in \ZZ_+}\gamma_m (A),
\end{equation}
where $\gamma_m(A)\in S(\CC[\nabla]\otimes \bar{\g}^f)\otimes (\CC[\nabla]\otimes \overline{[e,\g]})^{\otimes m}$. If we require that each $\w_i \in \WW^k(\g,f)$ satisfies the additional relation $\gamma_0(w_i)=\bar{\tilde{v}}^i_{(0)}$, then such $w_i$ is unique. This uniqueness of generators allows us to consider the differential superalgebra isomorphism
\begin{equation} \label{eq:omega}
  \w : S(\CC[\nabla]\otimes \bar{\g}^f) \to \WW^k(\bar{\g}, f), \quad \bar{\tilde{v}}^i_{(0)} \mapsto \w(\bar{\tilde{v}}^i_{(0)})= \w_i,
\end{equation}
where $\w_i$'s are the unique generators of the $W$-algebra, determined by Proposition \ref{prop:W-generator} and the extra property $\gamma_0(w_i)=\bar{\tilde{v}}^i_{(0)}$.

\subsection{SUSY Poisson structure of $\WW^k(\bar{\g},f)$} \label{subsec:W-bracket}

Consider $\g$ as an $\osp(1|2)$-module via the adjoint action of $\mathfrak{s}$. Then $\g$ can be decomposed into a sum of irreducible components $\g=\bigoplus_{i\in {\cI} }R_{\epsilon_i}$, where $R_{\epsilon_i}$ is the irreducible $\osp(1|2)$-module of dimension $2\epsilon_i+1$.
Denote $d_i:= 2 \epsilon_i$ and take a basis $ \{v_i^{(m)} | m=0,1, \cdots, d_i\}$  of $R_{\epsilon_i}$ which is homogeneous with respect to both the grading \eqref{grading} and the parity,
and satisfies 
\begin{equation} \label{eq:properties of basis}
  [f, v_i^{(m)}]=v_i^{(m+1)}
\end{equation}
for $m=0,1,2,\cdots, d_i-1$. 
Then 
\begin{enumerate}[(i)]
\item $\{v_i^{(d_i)}\}_{i\in {\cI}}$ and $\{v_i^{(0)}\}_{i\in {\cI}}$ are homogeneous bases of $\g^f$ and $\g^e:= \ker (\ad e)$, respectively,
\item $\mathcal{B}=\{ v_i^{(m)}| (i,m)\in {\cJ}\}$ is a basis of $\g$, where  
 $$  {\cJ}= \{ (i,m)| i\in {\cI}, \ m=0,1,\cdots, d_i\}.$$
\end{enumerate}

\noindent 
Now, we consider two other  bases $\{\tilde{v}^i_{(d_i)}|i\in {\cI}\}$ and $\mathcal{B}^\perp=\{ \tilde{v}^i_{(m)}| (i,m)\in {\cJ}\}$ of $\g^e$ and $\g$,
given by  
\begin{equation}\label{eq:dual bases}
(v_i^{(m)}|\tilde{v}^j_{(n)})=\delta_{i,j}\delta_{m,n}.
\end{equation}

\begin{remark}
  One can find such bases $\mathcal{B}$ and $\mathcal{B}^\perp$ of  $\g$ with all the desired properties by following section VI in \cite{Suh20}. More precisely, our $v_i^{(0)}$ and $\tilde{v}^i_{(0)}$ can be identified with $r^i$ and $r_i$ in \cite{Suh20}.
\end{remark}

Recall the map $\w$ in \eqref{eq:omega} and consider the projection map on lowest weight vectors:
\begin{equation}
\sharp: \g \to \g^f\,,\quad v_i^{(m)}\mapsto \delta_{m,d_i}\,v_i^{(d_i)},\text{ for }(i,m)\in {\cJ}. 
\end{equation}
We simply denote $g^\sharp:= \sharp(g)$ for $g\in \g$.

\begin{theorem} \cite{Suh20}\label{thm:w bracket}
Let $a\in \g^f_{-t_1}$ and $b\in \g^f_{-t_2}$, where $\g^f_N:= \g^f \cap \g(N)$ for $N\in \frac{1}{2}\ZZ$. Then 
\begin{equation} \label{eq:theorem, SUSY bracket}
  \begin{aligned}
     \{ \w(\bar{a}){}_\Lambda \w(\bar{b})\}& =(-1)^{p(a)}   (\w(\overline{[a,b]})+k(a|b)\chi) \\
     -(-1)^{p(a)p(b)+p(a)}& \sum_{p\in \ZZ_+} \sum_{\substack{-t_2-\frac{1}{2}<(i_0,m_0)< \cdots \\ \cdots <(i_p, m_p)< t_1}}
     \Big( \w(\overline{[b,v_{i_0}^{(m_0)}]}^\sharp)-k(b|v_{i_0}^{(m_0)})(D+\chi) \Big)\\
    & \times\bigg[ \prod_{t=1}^{\substack{\longrightarrow \\ p}} \Big( \w(\overline{[\tilde{v}^{i_{t-1}}_{(m_{t-1}+1)},v_{i_t}^{(m_t)}]}^\sharp) -k(\tilde{v}^{i_{t-1}}_{(m_{t-1}+1)}|v_{i_t}^{(m_t)})(D+\chi)\Big) \bigg]\\
    & \times \Big( \w(\overline{[\tilde{v}^{i_{p}}_{(m_{p}+1)},a]}^\sharp) -k(\tilde{v}^{i_{p}}_{(m_{p}+1)}|a)\chi\Big),
  \end{aligned}
\end{equation}
where $<$ is defined as follows: 
\begin{equation}
\begin{aligned}
& {\rm(i)}\ (i_{t-1}, m_{t-1})<(i_t, m_t) \text{ if and only if } 
\left\{
\begin{array}{l}
  \tilde{v}_{(m_{t-1})}^{i_{t-1}}\in \g(j_{t-1}),  \ \tilde{v}_{(m_{t})}^{i_{t}}\in \g(j_{t}),\\
  j_{t-1}<j_t,
\end{array}
\right.\\
& {\rm(ii)}\ -t_2-\frac{1}{2}<(i_0,m_0)
\text{ if and only if } 
  \tilde{v}_{(m_{0})}^{i_{0}}\in \g(j_{0}) \text{ and }
  -t_2\leq j_0,\\
& {\rm(iii)}\ (i_p,m_p)< t_1
\text{ if and only if } 
  \tilde{v}_{(m_{p})}^{i_{p}}\in \g(j_{p}) \text{ and }
  j_p<t_1,
\end{aligned} 
\end{equation}
and the ordered product $\displaystyle \prod_{t=1}^{\atopn{\longrightarrow}{p}} $  is defined by 
$\displaystyle \prod_{t=1}^{\atopn{\longrightarrow}{p}} \Omega_t:= \Omega_1 \Omega_2 \cdots \Omega_p$.
\end{theorem}

\begin{example}
Let $\g= \sll(2|1)$ with the basis elements $E,\fh ,F,e,f, \tilde{e}, \tilde{f}, U$ of $\g$ as in section \ref{sec: sl(2|1) in detail}. Then 
$\g=R_1 \oplus R_2$, where $R_1=\text{Span}_{\CC}\{\tilde{e}, \tilde{f}, U\}$ and $R_2=\text{Span}_{\CC}\{E,\fh ,F,e,f\}$. Consider the bases $\{v_1^{(0)},v_1^{(1)},v_1^{(2)}\}$ and $\{v_{2}^{(0)},v_{2}^{(1)},v_{2}^{(2)},v_{2}^{(3)},v_{2}^{(4)}\}$ of $R_1$ and $R_2$:
\begin{equation} \label{eq: sl(2|1) B basis}
\begin{aligned}
& v_1^{(2)}= -\tilde{f}, \quad v_1^{(1)}= -U, \quad v_1^{(0)}=-\tilde{e},\\
& v_2^{(4)}=-2F, \quad v_2^{(3)}= f, \quad v_2^{(2)}=2\fh , \quad v_2^{(1)}=-e, \quad v_2^{(0)}=E.
\end{aligned}
\end{equation}
The dual bases of $R_1$ and $R_2$ are:
\begin{equation} \label{eq: sl(2|1) B perp}
  \begin{aligned}
  & \textstyle  \tilde{v}_{(2)}^1= -\frac{1}{2}\tilde{e}, \quad \tilde{v}_{(1)}^1= \frac{U}{2}, \quad \tilde{v}_{(0)}^1=\frac{1}{2}\tilde{f},\\
  &  \textstyle \tilde{v}_{(4)}^{2}=-\frac{1}{2}E, \quad \tilde{v}_{(3)}^{2}= \frac{1}{2}e, \quad \tilde{v}_{(2)}^{2}=\fh , \quad \tilde{v}_{(1)}^{2}=\frac{1}{2}f, \quad \tilde{v}_{(0)}^{2}=F.
  \end{aligned}
  \end{equation}
The $N=1$ SUSY $W$-algebra $\WW^k(\overline{\sll}(2|1),f)$ is generated by $\w(\bar{v}_1^{(2)})$ and $\w(\bar{v}_{2}^{(4)})$ and $\Lambda$-brackets between them are as follows:
\begin{equation} \label{eq:W(sl(2|1))bracket}
\begin{aligned}
& \{\w(\bar{v}_1^{(2)}){}_\Lambda w(\bar{v}_1^{(2)})\}= \w(\bar{v}_{2}^{(4)})+2k^3 \lambda \chi, \\
& \{\w(\bar{v}_{2}^{(4)}){}_\Lambda w(\bar{v}_1^{(2)})\}= k^2(2\partial+\chi D + 2\lambda)\w(\bar{v}_{1}^{(2)}), \\
& \{\w(\bar{v}_{2}^{(4)}){}_\Lambda w(\bar{v}_{2}^{(4)})\}=k^2 (2\partial+\chi D + 3\lambda)\w(\bar{v}_{2}^{(4)}) + 2k^5\lambda^2 \chi.
\end{aligned}
\end{equation}
For example, we list below the nontrivial terms in formula of Theorem \ref{thm:w bracket} when $a=v_{2}^{(4)}\in\g^f_{-1}$ and $b=v_1^{(2)}\in\g^f_{-\frac12}$ to show how to get the second equality in \eqref{eq:W(sl(2|1))bracket}:
\begin{equation} \label{eq: N=2 bracket of tilde{f} and F}
\begin{aligned}
  \{\w(\bar{v}_{2}^{(4)}){}_\Lambda w(\bar{v}_1^{(2)})\}
  & = \w(\overline{[v_1^{(2)}, v_{2}^{(2)}]})\big(-k(\chi+D)(\tilde{v}_{(3)}^{2}|v_{2}^{(3)})\big)(-k\chi(\tilde{v}_{(4)}^{2}|v_{2}^{(4)}))\\
  & +\big(- k(\chi+D)(v_1^{(2)}| v_1^{(0)})\big)(\w(\overline{[\tilde{v}_{(1)}^1, v_{2}^{(3)}]}))(-k\chi(\tilde{v}^{2}_{(4)}|v_{2}^{(4)})) \\
  &+\big(-k(\chi+D)(v_1^{(2)}|v_1^{(0)})\big)\big(-k(\chi+D)(\tilde{v}_{(1)}^1|v_1^{(1)})\big)(\w(\overline{[\tilde{v}_{(2)}^1,v_2^{(4)}]})).
\end{aligned}
\end{equation}
\end{example}

\subsection{Poisson structure on $\WW^k(\g, F)$}

In this section, we briefly review on a non-SUSY $W$-algebra and the corresponding non-SUSY version of Theorem \ref{thm:w bracket}. (\cite{Suh20})

Recall that $V^k(\g)$ is a classical affine PVA of level $k$ introduced in Example \ref{ex:affine}. Let 
\begin{equation} \label{rho}
 \rho: \CC[\lambda]\otimes V^k(\g) \to \CC[\lambda]\otimes S(\CC[\partial]\otimes \mathfrak{q})
\end{equation}
be a differential algebra homomorphism 
such that $\rho(\lambda)=\lambda$,  $\rho(n)=(F|n)$ and $\rho(q)=q$ for $n\in \n$ and $q\in \mathfrak{q}$, where $\mathfrak{n}$ and $\mathfrak{q}$ are defined in \eqref{parabol}.  
The classical $W$-algebra associated with $\g$ and $f$ of level $k$ is 
\begin{equation} \label{def:SUSY W2}
 \textstyle  \mathcal{W}^k(\g,F)= \{ w\in S(\CC[\partial]\otimes \mathfrak{q}) \, | \, \rho(\ad \, n (w))=0 \text{ for } n\in \n\}.
\end{equation}
Here $\ad\, n (w) = \{\, n{}_\lambda w\}$ for the $\lambda$-bracket on $V^k(\g)$
and $\WW^k(\g,F)$ is the PVA via the bracket inherited from $V^k(\g)$.

The {\it conformal weight} $\Delta$ on $\WW^k(\g,F)$  
is induced from the gradation on $S(\CC[\partial]\otimes\mathfrak{q})$ such that 
\begin{equation} \label{eq:conformal weight}
  \Delta(a)=1-j_a, \quad \Delta(AB)= \Delta(A)+ \Delta(B), \quad \Delta(\partial(A))= 1+\Delta(A)
\end{equation}
for $a\in \g(j_a)$ and $A,B \in S(\CC[\partial]\otimes \mathfrak{q})$.

\begin{proposition}  \label{prop:nonSUSY W-generator} 
  Fix a basis $\{q_i^{(r_i)}|i\in {\cI}\}\subset \g$ of $\g^F$ whose elements are homogeneous with respect to the grading \eqref{grading} and the parity. 
Then for each $i\in {\cI}$, there is an element $\nu_i\in \WW^k(\g,F)$ whose linear part without total derivatives is $q_i^{(r_i)}$ and  conformal weight is
$\Delta(\nu_i)= \Delta(q_i^{(r_i)})$.
Moreover, such $\nu_i$'s for $i\in {\cI}$ generate $\WW^k(\g,F)$ as a differential algebra.
\end{proposition}

By the $\sll_2$-representation theory, we have a decomposition
\begin{equation} \label{eq: decomp g^F}
  \g^F \oplus[E,\g]= \g,
\end{equation}
which induces a natural projection map $\pi_0:S(\CC[\partial]\otimes \g) \to S(\CC[\partial]\otimes \g^F)$. As in $N=1$ SUSY case, one can show the uniqueness of the elements $\nu_i$'s in  $\WW^k(\g,F)$ satisfying Proposition \ref{prop:nonSUSY W-generator} and $\pi_0(\nu_i)=q_i^{(d_i)}$, which in turn gives a well-defined
 differential superalgebra isomorphism
\begin{equation} \label{eq:nu}
  \nu : S(\CC[\partial]\otimes \g^F) \to \WW^k(\g, F), \quad q_i^{(r_i)} \mapsto \nu_i.
\end{equation}

Take homogeneous bases $\{q_i^{(r_i)}|i\in {\cI}\}$ and $\{q_i^{(0)}|i\in {\cI}\}$ of $\g^F$ and $\g^E:= \ker (\ad E)$, respectively, such that the set $\mathcal{C}=\{ q_i^{(m)}| (i,m)\in {\cJ}\}$ forms a basis of $\g$, where ${\cJ}= \{ (i,m)| i\in {\cI}, \ m=0,1,\cdots, r_i\}$, and $q_i^{(m)}:=[F, q_i^{(m-1)}]$ for $m=1, \cdots , r_i$, where $2r_i+1$ is the dimension of the irreducible $\mathfrak{sl}(2)$-submodule of $\g$ containing $q_i ^{(r_i)}$. In addition, we consider its dual basis $\mathcal{C}^\perp=\{ \tilde{q}^i_{(m)}| (i,m)\in {\cJ}\}$ of $\g$ given by 
\begin{equation}\label{eq:dual bases-q}
(q_i^{(m)}|\tilde{q}^j_{(n)})=\delta_{i,j}\delta_{m,n}.
\end{equation}

\begin{theorem}\label{thm:nonSUSY w bracket} \cite{Suh20}
Let $\natural: \g \mapsto \g^F$ be the projection map and denote $g^\natural:= \natural(g)$ for $g\in \g$.
Consider $a\in \g^F_{-t_1}$ and $b\in \g^F_{-t_2}$, where $\g^F_N:=\g^F\cap \g(N)$. Then we have
\begin{equation} \label{eq:theorem, nonSUSY bracket}
  \begin{aligned}
     \{ \nu(a){}_\lambda \nu(b)\}& =  \nu([a,b])+k(a|b)\lambda \\
      -(-1)^{p(a)p(b)} & \sum_{p\in \ZZ_+} 
     \sum_{\substack{-t_2-1\prec(i_0,m_0)\prec \cdots \\ \cdots \prec(i_p, m_p)\prec t_1}}(-1)^{i_0}\Big( \nu([b,q_{i_0}^{(m_0)}]^\natural)-k(b|q_{i_0}^{(m_0)})(\partial+\lambda) \Big)\\
    & \times\bigg[ \prod_{t=1}^{\atopn{\longrightarrow}{p}} (-1)^{p(q_{i_t}^{(m_t)})}
    \Big( \nu([\tilde{q}^{i_{t-1}}_{(m_{t-1}+1)},q_{i_t}^{(m_t)}]^\natural) -k(\tilde{q}^{i_{t-1}}_{(m_{t-1}+1)}|q_{i_t}^{(m_t)})(\partial+\lambda)\Big) \bigg]\\
    & \times  \Big( \nu([\tilde{q}^{i_{p}}_{(m_{p}+1)},a]^\natural) -k(\tilde{q}^{i_{p}}_{(m_{p}+1)}|a)\lambda\Big),
  \end{aligned}
\end{equation}
where the partial order $\prec$ is defined as follows: 
\begin{equation}
\begin{aligned}
& {\rm(i)}\ (i_{t-1}, m_{t-1})\prec(i_t, m_t) \text{ if and only if } 
\left\{
\begin{array}{l}
  \tilde{q}_{(m_{t-1})}^{i_{t-1}}\in \g(j_{t-1}),  \ \tilde{q}_{(m_{t})}^{i_{t}}\in \g(j_{t}),\\
  j_{t-1}+1\leq j_t,
\end{array}
\right.\\
& {\rm(ii)}\ -t_2-1\prec(i_0,m_0)
\text{ if and only if } 
  \tilde{q}_{(m_{0})}^{i_{0}}\in \g(j_{0}) \text{ and }
  -t_2\leq j_0,\\
& {\rm(iii)}\ (i_p,m_p)\prec t_1
\text{ if and only if } 
  \tilde{q}_{m_{p}}^{i_{p}}\in \g(j_{p}) \text{ and }
  j_p+1\leq t_1.
\end{aligned} 
\end{equation}
\end{theorem}

Before ending this section, we show that, as in $N=1$ SUSY case, the conformal weight $\Delta$ on $\WW^k(\g,F)$ is indeed the weight decomposition (2) of Definition \ref{def: superconformal vector}, when $k\neq 0$. However, in the this case, we also have the primary generators.

\begin{theorem}[\cite{work2}] \label{thm:conformal in N=0 W}
  Let $\{q_i^{(0)}|i\in {\cI}_0\}$ and $\{\tilde{q}_{(0)}^i|i\in {\cI}_0\}$ be the bases of $\g^F\cap \g(0)$ given by $(q_i^{(0)}|\tilde{q}^j_{(0)})=\delta_{ij}$. If $k\neq 0$, then 
  \[\frac{2}{k}\bigg(\nu(F) +\frac{1}{2}\sum_{i\in {\cI}_0} \nu(q_i^{(0)})\nu(\tilde{q}^i_{(0)})\bigg)\]
  is a conformal vector of $\WW^k(\g,F)$, giving the conformal weight decomposition as in \eqref{eq:conformal weight}. Furthermore, each $\nu(q_i^{(r_i)})$ for $i\in {\cI}$ such that $(E|q_i^{(r_i)})=0$ is primary of conformal weight $\Delta(q_i^{(r_i)})$.
\end{theorem}

\section{Additional SUSY structures on $\WW^k(\overline{\sll}(n+1|n),f)$ and $\WW^k(\sll(n+1|n), F)$}
\label{sec: additional SUSY structure on W-algebras}

In this section, let $\g= \sll(n+1|n)$ and $f$ be the odd principal nilpotent. Note that the conditions (sl-1)--(sl-4) in section \ref{sec: sl(2|1) in detail} completely determines the principal $\sll(2|1)$ subalgebra $R_1 \oplus R_2 \subset \g$ and its basis elements $\fh ,E,F,e,f,\tilde{e},\tilde{f},U$. Additionally, assume that $k$ is a nonzero constant in $\CC$.

\subsection{N=2 SUSY structure on $\WW^k(\overline{\sll}(n+1|n), f)$}

Recall that in Theorem \ref{thm: N=1, N=2 pva relation with superconformal}, we found a way to find a $N=2$ superconformal vector along with an additional SUSY structure in $N=1$ SUSY PVA. In this section, we use this theorem to get Theorem \ref{thm:N=1 to N=2, sl(n+1|n)}, one of the main results of this paper. We show how to extend a $N=0$ or $N=1$ SUSY $W$-algebra to a $N=1$ or $N=2$ structure. The set $W_{\fk}$ of primary elements of the starting $N=0$ (resp. $N=1$) $W$-algebra is related to the set of highest weight vectors of $\fk=\sll_2$ (resp. $\fk=\osp(1|2)$) in $\g$. The primary elements of the enlarged SUSY structure will be a subset of $W_{\fk}$.

\begin{proposition} \label{prop:N=1 to N=2, sl(n+1|n)}
  Recall the isomorphism $\w$ in \eqref{eq:omega}. We have 
  \begin{enumerate}
    \item  $\{\w(\bar{\tilde{f}}){}_{\Lambda}\w(\bar{\tilde{f}})\}=-2\w(\bar{F})+2k^3 \lambda \chi$
    \item  $G:=-\frac{2}{k^2} \w(\bar{F})=k^2 \w(\bar{v}_2^{(4)})$ is a superconformal vector of $\WW^k(\bar{\g},f)$.
  \end{enumerate}
  Moreover, we have a set of $G$-primary elements $W_{\osp(1|2)}^{N=1}=\{\w(\bar{v}_j^{(2j)})\,|\,j=1,3,4, \cdots, 2n\}$  such that $W_{\osp(1|2)}^{N=1}\cup \{G\}$ freely generates $\WW^k(\bar{\g}, f)$ as a $\CC[\nabla]$-algebra.
\end{proposition}
\begin{proof}
We give the proof in section \ref{Proof of proposition}.
\end{proof} 
Proposition \ref{prop:N=1 to N=2, sl(n+1|n)} shows that the conditions of Theorem \ref{thm: N=1, N=2 pva relation with superconformal} are verified.
Let us denote
\begin{equation} \label{eq:K and G}
  J=\frac{\sqrt{-1}}{k}\w(\bar{\tilde{f}}) \text{ and } G= -\frac{2}{ k^2} \w(\bar{F}).
\end{equation}
Then, $J$ and $G$ will satisfy the conditions in Theorem \ref{thm: N=1, N=2 pva relation with superconformal}, 
if we further find $G$-primary generators of $\WW^k(\overline{\sll}(n+1|n),f)$ satisfying the second equality of \eqref{eq: N=1 condition for N=2 in Thm}. This is done in the following theorem.

\begin{theorem}  \label{thm:N=1 to N=2, sl(n+1|n)}
  Let $\tilde{D}:=J_{(0|0)}$ for $J$ in \eqref{eq:K and G}. Then $\WW^k(\bar{\g},f)$ endowed with $\bm{\nabla}=(\partial, D, \tilde{D})$ is a $N=2$ SUSY PVA with a $N=2$ superconformal vector $J$. Furthermore, the set $W_{\osp(1|2)}^{N=2}=\{\w(\bar{v}_{2i-1}^{(4i-2)})\,|\, i=2,3, \cdots n\}$ consists of $J$-primary elements, and $W_{\osp(1|2)}^{N=2} \cup \{J\}$ freely generates $\WW^k(\bar{\g},f)$ as a $\CC[\bm{\nabla}]$-algebra.
\end{theorem}

\begin{proof}
See section \ref{sec:proof of N=2 primary decomp}.
\end{proof}

\subsection{SUSY structures of $\WW^k(\sll(n+1|n), F)$}
In this section, we describe SUSY PVA structures of $\WW^k(\g, F)$ for $\g=\sll(n+1|n)$.

\begin{proposition} \label{prop:N=0 to N=1,2, sl(n+1|n)}
  Let $\nu$ be the isomorphism in \eqref{eq:nu}. Then 
  \begin{enumerate}
    \item   $\nu(f)_{(0)} \nu(f)= -2\nu(F)+\frac{1}{2}(\nu(U))^2$, $\nu(\tilde{f})_{(0)}\nu(\tilde{f})= 2\nu(F)-\frac{1}{2}(\nu(U))^2$, and \\ $\nu(f)_{(0)}\nu(\tilde{f})=0;$
    \item  $L=\frac{2}{k}(\nu(F)-\frac{1}{4}(\nu(U))^2)$ is a conformal vector of $\WW^k(\g,F)$ and all the elements of the set $W_{\sll_2}^{N=0}:=\{\nu(v_{m}^{(2m-s)})\,|\, m=1,2, \cdots, 2n,\,  s=0,1\}\setminus \{\nu(v_2^{(4)})\}$ are $L$-primary.
  \end{enumerate}
\end{proposition}
\begin{proof}
By Theorem \ref{thm:conformal in N=0 W}, it is known that $(\nu(F)-\frac{1}{4}(\nu(U))^2)_{(0)}=k\partial$, (2) directly follows. For (1), we use 
Theorem \ref{thm:nonSUSY w bracket}. 
More precisely, if $a=b= f$ (resp.$a=b=\tilde{f}$) then the only nonzero terms without $\lambda$ in \eqref{eq:theorem, nonSUSY bracket} are $\nu([a,b])$ and 
\begin{equation}
\begin{aligned}
& \frac{1}{2}\nu([f,\tilde{e}])\nu([\tilde{e},f]) \quad \Big(\text{ resp.}-\frac{1}{2}\nu([\tilde{f},e])\nu([e,\tilde{f}]) \ \Big).
\end{aligned}
\end{equation}
If $a=f$ and $b= \tilde{f}$ then every term without $\lambda$ in \eqref{eq:theorem, nonSUSY bracket} is zero.
\end{proof}

Proposition \ref{prop:N=0 to N=1,2, sl(n+1|n)}  implies the second main theorem of the article.
\begin{theorem} \label{thm:N=0 to N=1,2, sl(n+1|n)}
Consider the two elements
\[ G=\sqrt{-\frac{1}{k}}\nu(f)\  \text{ and } \ \tilde{G}=\sqrt{\frac{1}{k}}\nu(\tilde{f})\]
in $\WW^k(\g, F)$ and let 
$D:= G_{(0)}$ and $\tilde{D}:=\tilde{G}_{(0)}$. Then $\WW^k(\sll(n+1|n),F)$ endowed with $\bm{\nabla}=(\partial, D, \tilde{D})$ is a $N=2$ SUSY PVA.
\end{theorem}
\begin{proof}
The theorem follows from Proposition \ref{prop: pva, N=1, N=2 pva relation} and Proposition \ref{prop:N=0 to N=1,2, sl(n+1|n)}.
\end{proof}

\begin{theorem} \label{thm:N=2 conformal,N=0 W-alg}
    Recall  $G$, $D=G_{(0)}$, $\tilde{G}$ and $\tilde{D}=\tilde{G}_{(0)}$  in Theorem \ref{thm:N=0 to N=1,2, sl(n+1|n)} and let 
    \[ J:=-\sqrt{-1}\, \nu(U)\]
     in $\WW^k(\g,F)$. Then $G$ and $\tilde{G}$ are $N=1$ superconformal vectors with respect to the induced $\Lambda$ and $\tilde{\Lambda}$ brackets, respectively, and 
      $J$ is a $N=2$ superconformal vector of $\WW^k(\g,F)$ with respect to the induced $\bm{\Lambda}$-bracket.
     Moreover, the following statements hold:
  \begin{enumerate}
    \item  The set $W_{\sll_2}^{N=1}=\{\nu(v_{j}^{(2j-1)})|j=1,3,4,\cdots, 2n\}$ consists of $G$-primary elements, and $W_{\sll_2}^{N=1}\cup \{G\}$ freely generates $\WW^k(\g,F)$ as a $\CC[\nabla]$-algebra,
    \item  The set $\widetilde{W}_{\sll_2}^{N=1}=\{ \nu(v_{2i'-1}^{(4i'-2)}), \nu(v_{2i-1}^{(4i-3)})|i'=2,3, \cdots ,n,\,  i=1,2,\cdots, n\}$ consists of $\tilde{G}$-primary elements and $\widetilde{W}_{\sll_2}^{N=1}\cup \tilde{G}$ freely generates $\WW^k(\g,F)$ as a $\CC[\tilde{\nabla}]$-algebra,
    \item  The set $W_{\sll_2}^{N=2}=\{\nu(v_{2i-1}^{(4i-3)})|i=2,\cdots, n\}$ consists of $J$-primary elements and $W_{\sll_2}^{N=2}\cup \{J\}$ freely generates $\WW^k(\g,F)$ as a $\CC[\bm{\nabla}]$-algebra.
  \end{enumerate}
\end{theorem}
\begin{proof}
See section \ref{sec:Proof of theorems N=0}.
\end{proof}

\begin{remark}
We remind that the highest weight vectors of $\sll_2$ (resp. $\osp(1|2)$; resp. $\sll(1|2)$) are given by $v_i^{(2i)},\ v_i^{(2i-1)}$, $i=1,2,\dots, 2n$ (resp. $v_i^{(2i)}$, $i=1,2,\dots, 2n$; resp.  $v_{2i}^{(4i)}$, $i=1,2,\dots, n$). For the second $\osp(1|2)$ superalgebra generated by $\tilde{e}$ and $\tilde{f}$, the highest weights are   $v_{2i}^{(4i)}$, $v_{2i}^{(4i-1)}$, $i=1,2,\dots, n$. Then, the correspondence between the sets $W_{\sll_2}^{N=0}$, $W_{\osp(1|2)}^{N=1}$ and $W_{\osp(1|2)}^{N=2}$ and  the highest weight vectors of $\sll_2$, $\osp(1|2)$ and $\sll(1|2)$ respectively is immediate. However, one may wonder how the sets $W_{\sll_2}^{N=\sn}$, $\sn>0$ are related to these  highest weight vectors. It is done in the following way.
For simplicity's sake, we note $\propto$ for equalities up to multiplication by non-zero numbers. We first consider the set  
$W_{\sll_2}^{N=1}$. Since $D=G_{(0)}\propto \nu(f)$, then $D(\nu(v_i^{(2i-1)}))\propto\nu(v_i^{(2i)})$ and we recover the highest weight vectors of $\osp(1|2)$. In the same way, for $\widetilde W_{\sll_2}^{N=1}$, $\tilde D=\tilde G_{(0)}\propto \nu(\tilde f)$, so that 
$\tilde D(\nu(v_{2i-1}^{(4i-2)}))\propto\nu(v_{2i}^{(4i)})$ and $\tilde D(\nu(v_{2i-1}^{(4i-3)}))\propto\nu(v_{2i}^{(4i-1)})$. 
The case $W_{\sll_2}^{N=2}$ is dealt similarly using both $D$ and $\tilde D$.
\end{remark}

\section{Open questions\label{sec:conclu}}
The results presented in the core of the paper concern mainly $N=2$ SUSY, but the propositions 
presented in section \ref{sec: PVA and SUSY PVA} deal with the chain  
$N=0\, \hookrightarrow\, N=1\, \hookrightarrow\, N=2$ of SUSY PVAs, which  
can be easily generalized to a $N=\sn-1\, \hookrightarrow\, N=\sn$ chain of SUSY PVAs. The proofs are essentially the same as for $N=2$, so that we don't reproduce them, and just present the propositions. 

In this section, $\Lambda$ denotes the tuple $(\lambda, \chi_1, \cdots, \chi_{\sn-1})$ and $\mathbf\Lambda$
the tuple $(\lambda, \chi_1, \cdots, \chi_{\sn})$. In the same way, we note 
${\nabla}=(\partial, D_1, \cdots, D_{\sn-1})$ and $\bm{\nabla}=(\partial, D_1, \cdots, {D_\sn})$.
We set
\begin{equation}
  \{a{}_{\bm{\Lambda}}b\}^-=(-1)^{\sn-1}\{D_\sn a{}_{\Lambda}b\}+\chi_\sn\{ a{}_{\Lambda}b\}\,.
\end{equation}

 \begin{proposition} \label{prop:N-pva, N-1 pva relation}
  Let $P$ be a $N=\sn-1$ SUSY PVA. The following two statements are equivalent:
  \begin{enumerate}
    \item  $(P, 1, \nabla, \{\cdot {}_{\bm\Lambda}\cdot\}^-, \cdot)$ is a $N=\sn$ SUSY PVA,
    \item  There exists an odd derivation $D_\sn:P\rightarrow P$ satisfying 
    \begin{equation} 
      D_\sn^2=\partial, \quad \ D_\sn\{a{}_\Lambda b\}=(-1)^{\sn-1}\{D_\sn a{}_\Lambda b\}+ (-1)^{p(a)+\sn-1}\{a{}_\Lambda D_\sn b\}
    \end{equation}    
    for any $a,b\in P$.
  \end{enumerate}
\end{proposition}
This proposition allows to consider a chain $N=0\, \hookrightarrow\, N=1\, \hookrightarrow \cdots \hookrightarrow\, N=\sn$ of SUSY PVAs, in the spirit of 
Proposition \ref{prop: pva, N=1, N=2 pva relation}.
\begin{theorem} \label{thm:N-superconformal}
 Let $P$ be a $N=\sn-1$ SUSY PVA with a $N=\sn-1$ superconformal vector $G$ of central charge $c\in \CC$. Suppose that there are primary elements $J$ and $a_i$, $i\in {\cI}$, such that $J$ is of parity $(-1)^{\sn}$,
 \begin{equation} 
  \{J{}_{\Lambda}J\}=(-1)^{\sn-1}\Big(G+\frac{\lambda^{3-\sn}\chi_1\chi_2\cdots \chi_{\sn-1}}{3}c\Big),\quad \{J{}_{\Lambda}a_i\}=J_{(0|0\cdots0)}a_i,
 \end{equation}
 where $c=0$ when $\sn\geq 4$,
 and the set $\{J, G\}\cup \{a_i, J_{(0|0\cdots0)}a_i\}$ freely generates $P$ as a $\CC[\nabla]$-algebra. Then, for $D_{\sn}:=J_{(0|0\cdots0)}$, $P$ is a $N=\sn$ SUSY PVA with a $N=\sn$ superconformal vector $J$. Moreover, $P$ is freely generated by the set $\{J\}\cup \{a_i\}_{i\in {\cI}}$ as a $\CC[\bm{\nabla}]$-algebra, where $\{a_i\}_{i\in {\cI}}$ consists of $J$-primary elements.
\end{theorem}

Remark however that the application to $N=\sn$ SUSY $W$-algebras remains to be done.
We note that superconformal algebras with $\sn$ supersymmetry have being defined in physics in the 70's \cite{Ademollo}, in relation with the $\mathfrak{o}(\sn)$ algebras. It has to be related with the principal embeddings of $\mathfrak{osp}(0|2)\simeq \mathfrak{sl}(2)$, $\mathfrak{osp}(1|2)$ and $\mathfrak{osp}(2|2)\simeq \mathfrak{sl}(1|2)$ in Lie superalgebras for $N=0,1,2$ SUSY respectively. 
In view of these results, it is tempting to conjecture that $N=\sn$ SUSY $W$-algebras should be related to $\mathfrak{osp}(\sn|2)$ embeddings into superalgebras. However, the notion of principal embeddings of $\mathfrak{osp}(\sn|2)$, $n>2$ has to be precised. As a first step, we remark that only the superalgebras 
$\mathfrak{osp}(\sn|2)$ with $\sn=1,2,3,4$ admit a principal $\mathfrak{osp}(1|2)$, so that these later cases appear to be more natural. Looking for $\mathfrak{sl}(1|2)$ principal embeddings, one is restricted to $\mathfrak{sl}(n+1|n)$. 
Moreover we note that the $\mathfrak{osp}(3|2)$ superalgebra can be viewed as the folding of a 
$\mathfrak{sl}(3|2)$ superalgebra, in the same way $\mathfrak{osp}(1|2)$ can be viewed as the folding of $\mathfrak{sl}(1|2)$ \cite{dico}. Then, embeddings of $\mathfrak{sl}(3|2)$ superalgebras in $\mathfrak{sl}(n+1|n)$ looks like the more promising direction to study $N=3$ SUSY $W$-algebras. Yet, a
detailed analysis remains to be done.

We also observe that the Sugawara and Kac-Todorov constructions recalled in Examples \ref{ex:affine} and \ref{ex: susy affine pva} are not fully understood at the level of $N=2$ SUSY affine VA or PVA. There was some attempts in the physics literature \cite{RS90,FOFR,delius}, based on the so-called coset construction, or using constrained $N=2$ superfields, but a comprehensive study is still needed.

Finally, we note that the $N=2$ SUSY $W$-algebras presented in this paper 
correspond to principal embeddings of 
$\mathfrak{sl}(1|2)$ into Lie superalgebras. In fact, any embedding of $\mathfrak{sl}(1|2)$ should 
lead to $N=2$ SUSY $W$-algebras, see \cite{RSS96}. A detailed analysis of this more general case 
will be the subject of a future work \cite{work}.

\medskip

\textbf{Acknowledgements}

E. Ragoucy warmly thanks the Seoul National University for partial support, and for the kind hospitality when part of this work was done.

\appendix
\section{Proofs of Propositions and Theorems}

We consider the decomposition  $\g= \bigoplus_{i=1}^{2n} R_i$ of $\g$ into  irreducible $\osp(1|2)$-module $R_i$ of dimension $2i+1$ together
with the dual bases $\mathcal{B}= \{v_i^{(m)}|(i,m)\in {\cJ} \}$ and  $\mathcal{B}^\perp= \{\tilde{v}^i_{(m)}|(i,m)\in {\cJ}\}$ of $\g$ introduced in section \ref{subsec:W-bracket}, where ${\cJ}= \{(i,m)|i=1,2,\cdots, 2n, \, m=0,1,\cdots, 2i\}$. Suppose $E=v_2^{(0)}\in\mathcal{B}$ and  $E,F,\fh ,f, e, \tilde{f}, \tilde{e}, U \in R_1 \oplus R_2 \subset \g$ satisfy (sl-1)--(sl-4) in section \ref{sec: sl(2|1) in detail}. We can assume
\[ [\tilde{f}, v_{2i}^{(0)}]=-v_{2i-1}^{(0)}\]
for any $i=1,2,\cdots, n$. 
Then, from $[\tilde{f}, f]=0$ and $(\text{ad}\tilde{f})^2=-(\text{ad}f)^2$, we get 
\begin{equation}\label{eq:action-f-tilde}
[\tilde{f}, v_{2i-1}^{(t)}]= (-1)^t\,v_{2i}^{(t+2)}, \quad [\tilde{f}, v_{2i}^{(s)}]=(-1)^{s+1}\,v_{2i-1}^{(s)}\,,
\end{equation}
for $t=0,1,\cdots, 2i-2$ and $s= 2,3,\cdots, 2i$. In addition, since $ \tilde{v}^\al_{(2\al)}$ is of parity $\alpha$ (mod 2) and belongs to  the one dimensional space  $\g^e \cap \g(\alpha)= \CC\cdot \left(\sum_{\be=1}^{2n-\al+1}e_{\be, \al+\be}\right)$, we can show $[\tilde{v}^i_{(2i)}, \tilde{v}^j_{(2j)}]= 0$
if $i$ or $j$ is even. Finally, we remind that
\begin{equation}
v_i^{(m)}\in\g \Big(\frac{i-m}2 \Big)\ \text{ and }\ \tilde{v}^i_{(m)}\in\g \Big(\frac{m-i}2 \Big)\,.
\end{equation}

\pagebreak

\subsection{Proof of Proposition \ref{prop:N=1 to N=2, sl(n+1|n)}} \label{Proof of proposition} 
\subsubsection{Proof of (1) in Proposition \ref{prop:N=1 to N=2, sl(n+1|n)}.} 
For such a purpose, let us compute $[\w(\bar{\tilde{f}}){}_\Lambda \w(\bar{\tilde{f}})]$.  

\begin{lemma} \label{lem:N=2 tilde{f} and tilde{f}}
For $a=b= \tilde{f}\in\g^{f}_{-1/2}$, the only nontrivial terms in \eqref{eq:theorem, SUSY bracket} are $ -\w(\overline{[a,b]})$ and 
  \begin{equation} \label{eq:N=2 tilde{f} and tilde{f}}
  - \big(k(b|v_{1}^{(0)})(D+\chi) \big)\big(k(\tilde{v}^{1}_{(1)}|v_{1}^{(1)})(D+\chi) \big)\big(k(\tilde{v}^{1}_{(2)}|a)(D+\chi) \big)=2k^3\lambda\chi.
  \end{equation}
\end{lemma}
\begin{proof}
  Let us show the only nonzero terms in $[\w(\bar{a}){}_\Lambda \w(\bar{b})]+\w(\overline{[a,b]})$ are given by  \eqref{eq:N=2 tilde{f} and tilde{f}}. 
  By relation \eqref{eq:theorem, SUSY bracket}, we have $v_{i_0}^{(m_0)}\in \g(j)$ for $j=-1, -\frac{1}{2}, 0, \frac{1}{2}.$
  \begin{enumerate}[(i)]
    \item Let $v_{i_0}^{(m_0)}\in \g(-1)$. Then $(b|v_{i_0}^{(m_0)})=0$. Since $[b,v_{i_0}^{(m_0)}]\in\g(-\frac32)$, then $[b,v_{i_0}^{(m_0)}]^\sharp\neq 0$ implies  $[\tilde{f},v_{i_0}^{(m_0)}] \in \CC v_3^0$ (see Lemma \ref{lemma:sl repn}). This in turn shows that $\tilde{v}^{i_t}_{(m_t+1)} \in \bigoplus_{j\geq \frac{3}{2}}\g(j)$ for $t\in \ZZ_{+}$, so that $(\tilde{v}^{i_t}_{(m_t+1)}| a)=0$. From $\g^{f}\cap \n=\emptyset$, we also have $[\tilde{v}^{i_t}_{(m_t+1)}, a]^\sharp=0$. Hence $v_{i_0}^{(m_0)}\in \g(-1)$ cannot give rise to nonzero terms.
    \item Let $v_{i_0}^{(m_0)}\in \g(-\frac{1}{2})$. Then $(b|v_{i_0}^{(m_0)})=0$ and $[b,v_{i_0}^{(m_0)}]^\sharp\neq 0$ implies $v_{i_0}^{(m_0)} \in \CC \tilde{f}$. However, if  $v_{i_0}^{(m_0)} \in \CC \tilde{f}$ then $\tilde{v}^{i_0}_{(m_0+1)}=0$. Hence we cannot get any nonzero terms.
    \item Let $v_{i_0}^{(m_0)}\in \g(0)$. Then $(b|v_{i_0}^{(m_0)})=0$ and $[b,v_{i_0}^{(m_0)}]^\sharp\neq 0$ implies $v_{i_0}^{(m_0)} \in \CC \fh $ and $v^{i_0}_{(m_0+1)} \in \CC e$. In this case, $[\tilde{v}^{i_0}_{(m_0+1)}, g]^\sharp=0$ for any $g\in \g$,  since $[e,\g] \oplus \g^f=0$.
    \item Let $v_{i_0}^{(m_0)}\in \g(\frac{1}{2})$. The only nonzero term arises only when $v_{i_0}^{(m_0)}\in \CC\tilde{e}$ and $v_{i_1}^{(m_1)}\in \CC U$. In this case we get relation \eqref{eq:N=2 tilde{f} and tilde{f}}.
  \end{enumerate}
\end{proof}

By Lemma \ref{lem:N=2 tilde{f} and tilde{f}}, we have 
$ \{\w(\bar{\tilde{f}}{})_\Lambda \w(\bar{\tilde{f}})\}=-2\w(\bar{F})+2k^3 \lambda \chi$.
Hence we proved Proposition \ref{prop:N=1 to N=2, sl(n+1|n)} (1). 

\subsubsection{Additionnal lemmas.}
To prove Proposition \ref{prop:N=1 to N=2, sl(n+1|n)} (2), let us introduce some more lemmas.
We remind that ${\cI}= \{ 1,2,\cdots, 2n+1\}$ and ${\cJ}= \{ (i,m)| i\in {\cI}, \ m=0,1,\cdots, d_i\}$.

\begin{lemma}\label{eq:key lemma for brackets}
  Let $i, j \in {\cI}$ and $(j,t)\in {\cJ}$. Then 
  \begin{enumerate}
    \item  $[\tilde{v}^i_{(2i)}, v_j^{(t)}]^\sharp=0$ for $t \leq 2j-1$,
    \item  $[\tilde{v}^i_{(2i-1)}, v_j^{(t)}]^\sharp=0$ for $t \leq 2j-2$.
  \end{enumerate}
\end{lemma}
\begin{proof}
For (1), observe that $v_j^{(t)}$ for $t \leq 2j-1$ is in $[e,\g]$. Let $u_j^{(t+1)} \in \g$ satisfy $[e,u_j^{(t+1)}]= v_j^{(t)}$. Since $\tilde{v}_{(2i)}^i \in \g^e$, we have 
$ [e,[u_j^{(t+1)}, \tilde{v}_{(2i)}^i]]= [v_j^{(t)},\tilde{v}_{(2i)}^i ] \in [e,\g]. $
Hence, from $[e,\g]\cap \g^f=\emptyset$,  we proved (1).\\
For (2), let $t \leq 2j-2$. Observe there are ${u'}_j^{(t+2)} \in \g$ and $\tilde{u}_{(2i)}^i \in \g^e$ such that $v_j^{(t)}=[E, {u'}_j^{(t+2)}]$ and  $\tilde{v}^i_{(2i-1)}=[f, \tilde{u}_{(2i)}^i]$. 
Now we have 
$ [\tilde{v}^i_{(2i-1)}, v_j^{(t)}]^\sharp= [E,[[f,\tilde{u}_{(2i)}^i],{u'}_j^{(t+2)}]]^\sharp + [[[f,\tilde{u}_{(2i)}^i], E],{u'}_j^{(t+2)}]^\sharp.$
Since $[[f,\tilde{u}_{(2i)}^i],E]=0$ and $[E,[[f,\tilde{u}_{(2i)}^i],{u'}_j^{(t+2)}]]^\sharp=0$, we conclude $[\tilde{v}^i_{(2i-1)}, v_j^{(t)}]^\sharp=0$.
\end{proof}

Now, we can list nonzero terms in relation \eqref{eq:theorem, SUSY bracket} when $b=\tilde{f}$ or $F$. To lighten the notations, let us denote
\begin{equation} \label{eq: tilde of omega}
\widetilde{\w}(\overline{[x,y]}):= \w(\overline{[x,y]}^\sharp)-k(x|y)(\chi+D).
\end{equation}

\begin{lemma} \label{lem:N=1 bracket F and tilde f}
  Let $b=F$. For $a= \tilde{f}$, the bracket $[\w(b){}_{\Lambda}\w(a)]$ can be computed similarly to what has been done for \eqref{eq: N=2 bracket of tilde{f} and F}.
  If $a=F$, the only nonzero terms in \eqref{eq:theorem, SUSY bracket} can be  listed as follows:
  \begin{equation}\label{eq:A5}
    \begin{aligned}
    &\widetilde{\w}(\overline{[b,v_1^{(0)}]})\widetilde{\w}(\overline{[\tilde{v}^1_{(1)},v^{(1)}_1]})\widetilde{\w}(\overline{[\tilde{v}^1_{(2)},a]})= \frac{1}{2}k \,  \w(\bar{\tilde{f}}) (\chi+D)\w(\bar{\tilde{f}}),\\
    &\widetilde{\w}(\overline{[b,v_1^{(0)}]})\widetilde{\w}(\overline{[\tilde{v}^1_{(1)},v^{(3)}_2]})\widetilde{\w}(\overline{[\tilde{v}^2_{(4)},a]})= -\frac{1}{4} k \, \w(\bar{\tilde{f}})^2 \chi,\\
    &\widetilde{\w}(\overline{[b,v_2^{(0)}]})\widetilde{\w}(\overline{[\tilde{v}^2_{(1)},v^{(1)}_1]})\widetilde{\w}(\overline{[\tilde{v}^1_{(2)},a]})=-\frac{1}{4}k(\chi+D)\w(\bar{\tilde{f}})^2 ,\\
    &\widetilde{\w}(\overline{[b,v_2^{(0)}]})\widetilde{\w}(\overline{[\tilde{v}^2_{(1)},v^{(1)}_2]})\widetilde{\w}(\overline{[\tilde{v}^2_{(2)},v^{(2)}_2]})\widetilde{\w}(\overline{[\tilde{v}^2_{(3)},v^{(3)}_2]})\widetilde{\w}(\overline{[\tilde{v}^2_{(4)},a]})= \frac{1}{2}k^5\, (\lambda+\partial)^2(\chi+D),\\
    &\widetilde{\w}(\overline{[b,v_2^{(0)}]})\widetilde{\w}(\overline{[\tilde{v}^2_{(1)},v^{(1)}_2]})\widetilde{\w}(\overline{[\tilde{v}^2_{(2)},a]})= k^2(\lambda+\partial)\w(\bar{F}),\\
    &\widetilde{\w}(\overline{[b,v^{(2)}_2]})\widetilde{\w}(\overline{[\tilde{v}^2_{(3)},v^{(3)}_2]})\widetilde{\w}(\overline{[\tilde{v}^2_{(4)},a]})= k^2 \w(\bar{F})\lambda,\\
    &\widetilde{\w}(\overline{[b,v_2^{(0)}]})\widetilde{\w}(\overline{[\tilde{v}^2_{(1)},v^{(3)}_2]})\widetilde{\w}(\overline{[\tilde{v}^2_{(4)},a]})=\frac{1}{2}k^2 \, (\chi+D)\w(\bar{F})\chi.\\
    \end{aligned}
  \end{equation}
\end{lemma}
\begin{proof}
Use Lemma \ref{eq:key lemma for brackets} and track all the nonzero terms as in the proof of Lemma \ref{lem:N=2 tilde{f} and tilde{f}}. Essentially, 
the computations done for relation \eqref{eq:W(sl(2|1))bracket} work.
\end{proof}

By equation \eqref{eq: N=2 bracket of tilde{f} and F}, and adding the contributions in relation \eqref{eq:A5}, we get 
\begin{equation} \label{eq:bracket of G and f,G}
  \begin{aligned}
    & \{ G{}_\Lambda \w(\tilde{f}) \}= (2\partial + \chi D + 2\lambda) \w(\tilde{f})\\
    & \{ G{}_\Lambda G\} = (2\partial + \chi D + 3\lambda) G -2k \lambda^2\chi
  \end{aligned}
\end{equation}
for $G:= -\frac{2}{k^2}\w(\bar{F})$. Hence $\w(\tilde{f})$ is $G$-primary of conformal weight $1$.

\begin{lemma} \label{lem:F and v_3 and v_4}
  Let $i\geq 3$ be an integer.
 Then 
  \begin{enumerate}
    \item  $  [ \tilde{e}, v_{2j}^{(4j)}]=2j \,  v^{(4j-2)}_{2j-1}, $
    \item  $  [ \tilde{e}, v_{2j-1}^{(4j-2)}]=-(2j-1) v^{(4j-2)}_{2j}. $
  \end{enumerate}
\end{lemma}
\begin{proof}
  We know $[ \tilde{e}, v_{2j}^{(4j)}]= c\, v_{2j-1}^{(4j-2)}$ for a constant $c\in \CC$. Since  
  $[\tilde{f},[\tilde{e}, v_{2j}^{(4j)}]]=-[2\fh , v_{2j}^{(4j)}]= 2j  \, v_{2j}^{(4j)}$ and 
  $[\tilde{f},[\tilde{e}, v_{2j}^{(4j)}]]= c\,[\tilde{f}, v_{2j-1}^{(4j-2)}]= c\,  v_{2j}^{(4j)}$, Hence  $c=2j$ and (1) follows.\\
  The proof for (2) is analogous to the one for (1) once one uses equation \eqref{eq:action-f-tilde}.
\end{proof}

\begin{lemma} \label{lem:key lemma for brackets-2}
  Suppose $(j,m)\in \mathcal{J}$. Then the following properties hold.
  \begin{enumerate}
  \item If $\widetilde{w}(\overline{[F,v_j^{(m)}]})\neq 0 $, then $(j,m)= (2,0)$ or $(j, 2j-2)$.
  \item If $\widetilde{w}(\overline{[\tilde{v}^2_{(1)},v_j^{(m)}]})\neq 0 $, then $(j,m)=(j,2j-1)$ or $(j,m)=(2,1)$.
  \item If $\widetilde{w}(\overline{[\tilde{v}^i_{(2i-1)},v_j^{(m)}]})\neq 0 $ for $i\in \mathcal{I}$, then 
  \[(j,m)\in \{ (i',2i'-1), (i'',2i'') \in \mathcal{J} | i'\geq i, i''>i\}.\]
  \item If $\widetilde{w}(\overline{[\tilde{v}^i_{(2i)},v_j^{(m)}]})\neq 0 $ for $i\in \mathcal{I}$, Then
  \[(j,m)\in \{  (i',2i') \in \mathcal{J} | i'\geq i\}.\]
  \item Let $(i_0, m_0), (i_1, m_2), \cdots (i_p,m_p) \in \mathcal{J}$ and $a,b\in \g^f$. If 
   \[ \widetilde{w}(\overline{[b,v_{i_0}^{(m_0)}]}) \widetilde{w}(\overline{[\tilde{v}^{i_0}_{(m_0+1)},v_{i_1}^{(m_1)}]}) \cdots \widetilde{w}(\overline{[\tilde{v}^{i_p}_{(m_p+1)},a]})\neq 0,\] 
   then $m_t <2 i_t$ for $t=0,1,\cdots, p$. 
  \end{enumerate}
  \end{lemma}
\begin{proof}
 (1), (2) directly follow from the $\osp(1|2)$ representation theory and the fact that for $(i,m), (i',m')\in \mathcal{J}$ we have $(\tilde{v}^i_{(m)}|v_{i'}^{(m')})\neq 0$ iff $i=i'$ and $m=m'$. (5) holds since if $m_t\geq 2i_t$ then $\tilde{v}^{i_t}_{(m_t+1)}=0$. For (3), Lemma \ref{eq:key lemma for brackets} shows $\widetilde{w}(\overline{[\tilde{v}^i_{(2i-1)},v_j^{(m)}]})\neq 0 $ only when $m=2j-1$ or $2j$. Now, since $\g^f \subset \bigoplus_{t<0} \g(t)$, we need $j\geq i$. In addition, we have $\widetilde{w}(\overline{[\tilde{v}^i_{(2i-1)},v_i^{(2i)}]})=0$ since $[\tilde{v}^i_{(2i-1)},v_i^{(2i)}]^\sharp = (v_1^{(0)}|[\tilde{v}^i_{(2i-1)},v_i^{(2i)}])\tilde{v}^1_{(0)}$ and $(v_1^{(0)}|[\tilde{v}^i_{(2i-1)},v_i^{(2i)}])=(-1)^{i+1}(\tilde{v}^i_{(2i-1)}|[\tilde{e},v_i^{(2i)}])=0$ by Lemma \ref{lem:F and v_3 and v_4}.
 (4) is proved similarly.
\end{proof}

\begin{lemma} \label{lem:N=1 bracket for F in w(sl(3|2))}
  Suppose $ n,i\geq 2$ are integers and $i\leq n$. Take $a=v^{(4i-2)}_{2i-1}$ and $b=F=-\frac12\,v_2^{(4)}$ in $\g=\sll(n+1|n)$. 
  We restrict ourselves to the terms in \eqref{eq:theorem, SUSY bracket} such that
  \begin{equation} \label{eq:W(sl(3|2))assumption}
   i_0, i_1, \cdots, i_p\in\{1,2,2i-1, 2i\}\,.
  \end{equation}
  Then, the only possibly nonzero terms in  \eqref{eq:theorem, SUSY bracket} satisfying \eqref{eq:W(sl(3|2))assumption}
  are listed below:
      \begin{equation} \label{eq:N=1 bracket for F and v3 w(sl(3|2))}
        \begin{aligned}
            &\widetilde{\w}(\overline{[b,v^{(0)}_2]})\,\widetilde{\w}(\overline{[\tilde{v}^2_{(1)},v^{(1)}_2]})\,\widetilde{\w}(\overline{[\tilde{v}^2_{(2)},a]})=  \frac{2i-1}{2} k^2 \,(\lambda+\partial)\,  \w(\bar{v}_{2i-1}^{(4i-2)}),\\
            &\widetilde{\w}(\overline{[b,v^{(0)}_2]})\,\widetilde{\w}(\overline{[\tilde{v}^2_{(1)},v^{(4i-3)}_{2i-1}]})\,\widetilde{\w}(\overline{[\tilde{v}^{2i-1}_{(4i-2)},a]})=\frac{k^2}{2}(-\lambda+\chi D)\, \w(\bar{v}_{2i-1}^{(4i-2)}),\\
            &\widetilde{\w}(\overline{[b,v^{(4i-4)}_{2i-1}]})\,\widetilde{\w}(\overline{[\tilde{v}^{2i-1}_{(4i-3)},v^{(4i-3)}_{2i-1}]})\,\widetilde{\w}(\overline{[\tilde{v}^{2i-1}_{(4i-2)},a]})=k^2\, \w(\bar{v}_{2i-1}^{(4i-2)}) \lambda.
        \end{aligned}
      \end{equation}
\end{lemma}
\begin{proof}
Using Lemma \ref{lem:key lemma for brackets-2}, one shows that the only non-zero terms in \eqref{eq:theorem, SUSY bracket} are 
given by  \eqref{eq:N=1 bracket for F and v3 w(sl(3|2))}. 
For instance, one can show that 
there is no nonzero term starting with $\widetilde{\w}(\overline{[b,v^{(2)}_2]})$ in the following way.
Since $[\tilde{v}^2_{(3)}, \g]\oplus \g^f=\g$, we have  $\widetilde{\w}(\overline{[\tilde{v}^2_{(3)},v_i^{(m)}]}) \neq 0$ iff $(i,m)=(2,3).$ Now, take $(i,m)=(2,3)$ and consider $(i',m')$ such that  $\widetilde{\w}(\overline{[\tilde{v}^i_{(m+1)},v_{i'}^{(m')}]})\neq 0$. Then $(i',m')$ should be $(2,4)$. Finally, Lemma \ref{lem:key lemma for brackets-2} (5) tells that there is no nonzero term starting with $\widetilde{\w}(\overline{[b,v^{(2)}_2]})$. 
Furthermore, the equalities in \eqref{eq:N=1 bracket for F and v3 w(sl(3|2))} are obtained by direct computations.
\end{proof}

\begin{lemma} \label{lem:N=1 bracket for F in w(sl(3|2))-2}
  Suppose $n,i\geq 2$ are integers and $i\leq n$. Take $a=v^{(4i)}_{2i}$ and $b=F=-\frac12\,v_2^{(4)}$ in $\g=\sll(n+1|n)$. The only possibly nonzero terms in \eqref{eq:theorem, SUSY bracket} satisfying \eqref{eq:W(sl(3|2))assumption}
  are listed below:
      \begin{equation}\label{eq:N=1 bracket for F and v4 w(sl(3|2))}
        \begin{aligned}
        &\widetilde{\w}(\overline{[b,v_1^{(0)}]})\widetilde{\w}(\overline{[\tilde{v}^1_{(1)},v^{(1)}_1]})\widetilde{\w}(\overline{[\tilde{v}^1_{(2)},a]})=-ik\, \w(\bar{v}_1^{(2)})(\chi+D) \w(\bar{v}_{2i-1}^{(4i-2)}) ,\\
        &\widetilde{\w}(\overline{[b,v_1^{(0)}]})\widetilde{\w}(\overline{[\tilde{v}^1_{(1)},v^{(4i-1)}_{2i}]})\widetilde{\w}(\overline{[\tilde{v}^{2i}_{(4i)},a]})= \frac{1}{2}k\,\w(\bar{v}_1^{(2)})\w(\bar{v}_{2i-1}^{(4i-2)}) \chi,\\
        &\widetilde{\w}(\overline{[b,v_2^{(0)}]})\widetilde{\w}(\overline{[\tilde{v}^2_{(1)},v^{(1)}_1]})\widetilde{\w}(\overline{[\tilde{v}^1_{(2)},a]})=\frac{i}{2}k\, (\chi+D)\w(\bar{v}_1^{(2)})\w(\bar{v}_{2i-1}^{(4i-2)}),\\
        &\widetilde{\w}(\overline{[b,v_2^{(0)}]})\widetilde{\w}(\overline{[\tilde{v}^2_{(1)},v^{(4i-3)}_{2i-1}]})\widetilde{\w}(\overline{[\tilde{v}^{2i-1}_{(4i-2)},a]})= \frac{i}{2} k\, (\chi+D)\w(\bar{v}_{2i-1}^{(4i-2)}) \w(\bar{v}_1^{(2)}),\\
        &\widetilde{\w}(\overline{[b,v^{(4i-4)}_{2i-1}]})\widetilde{\w}(\overline{[\tilde{v}^{2i-1}_{(4i-3)},v^{(4i-3)}_{2i-1}]})\widetilde{\w}(\overline{[\tilde{v}^{2i-1}_{(4i-2)},a]})=-ik\,  \w(\bar{v}_{2i-1}^{(4i-2)}) (\chi+D) \w(\bar{v}_1^{(2)}),\\
        &\widetilde{\w}(\overline{[b,v^{(4i-4)}_{2i-1}]})\widetilde{\w}(\overline{[\tilde{v}^{2i-1}_{(4i-3)},v^{(4i-1)}_{2i}]})\widetilde{\w}(\overline{[\tilde{v}^{2i}_{(4i)},a]})=\frac{2i-1}{2}k  \, \w(\bar{v}_{2i-1}^{(4i-2)}) \w(\bar{v}_1^{(2)})\chi,\\
        &\widetilde{\w}(\overline{[b,v_2^{(0)}]})\widetilde{\w}(\overline{[\tilde{v}^2_{(1)},v^{(1)}_2]})\widetilde{\w}(\overline{[\tilde{v}^2_{(2)},a]})=i k^2 \, (\lambda+\partial) \w(\bar{v}^{(4i)}_{2i}),\\
        &\widetilde{\w}(\overline{[b,v_2^{(0)}]})\widetilde{\w}(\overline{[\tilde{v}^2_{(1)},v^{(4i-1)}_{2i}]})\widetilde{\w}(\overline{[\tilde{v}^{2i}_{(4i)},a]})= -\frac{k^2}{2}(\lambda-\chi D) \w(\bar{v}_{2i}^{(4i)}),\\
        &\widetilde{\w}(\overline{[b,v^{(4i-2)}_{2i}]})\widetilde{\w}(\overline{[\tilde{v}^{2i}_{(4i-1)},v^{(4i-1)}_{2i}]})\widetilde{\w}(\overline{[\tilde{v}^{2i}_{(4i)},a]})=k^2 \, \w(\bar{v}^{(4i)}_{2i})\lambda.\\
        \end{aligned}
      \end{equation}
\end{lemma}
\begin{proof}
The proof of the lemma is similar to the one of  Lemma \ref{lem:N=1 bracket for F in w(sl(3|2))}. We use Lemma \ref{lem:key lemma for brackets-2} and check the equalities by computations.
\end{proof}
\begin{corollary}
  We have 
\begin{equation} \label{eq:F and v_{2n-1} and v_{2n}}
  \begin{aligned}
    & [\w(\bar{v}^{(6)}_{3}){}_\Lambda \w(\bar{F})]= k^2 \bigg( 2\lambda + \frac{1}{2}\chi D + \frac{3}{2}\partial \bigg)\w(\bar{v}^{(6)}_{3}),\\
    & [\w(\bar{v}^{(8)}_{4}){}_\Lambda \w(\bar{F})]= -k^2 \bigg( \frac{5}{2}\lambda + \frac{1}{2}\chi D + 2\partial \bigg)\w(\bar{v}^{(8)}_{4}).
  \end{aligned}
\end{equation}
Hence for $n=2$, $G= -\frac{2}{k^2}\w(\bar{F})$ is a superconformal vector of $\WW^k(\overline{\sll}(3|2), f)$. Furthermore, when $n\geq 2$, 
the two elements $\w(\bar{v}^{(6)}_{3})$ and $\w(\bar{v}^{(8)}_{4})$ are $G$-primary of conformal weight $2$ and $\frac52$ respectively.
\end{corollary}
\begin{proof}
For the case  $\g= \sll(3|2)$, $i=2$ is the only possibility and \eqref{eq:W(sl(3|2))assumption} does not bring any further constrain.
 Hence it suffices to consider Lemma \ref{lem:N=1 bracket for F in w(sl(3|2))} and Lemma \ref{lem:N=1 bracket for F in w(sl(3|2))-2}.
\end{proof}

In order to compute $[\w(\bar{v}^{(2i)}_{i}){}_\Lambda \w(\bar{F})]$ for $i\geq 5,$ we have to find all nontrivial terms in relation \eqref{eq:theorem, SUSY bracket} when $b=F$ and $a= \bar{v}^{(2i)}_{i}$ (see Lemma \ref{lem:N=1 bracket for F in general} and \ref{lem:N=1 bracket for F in general-2}). We first use the following lemma and corollaries (see Lemma \ref{lem:N=1 bracket for F in general--}, Corollaries \ref{cor:key in lemma A.12-1} and \ref{cor:key in lemma A.12}).

\begin{lemma} \label{lem:N=1 bracket for F in general--}
  Let $j\in \ZZ_{+}$ be an integer such that $j\leq n$. We have the following properties:
  \begin{enumerate}
    \item  $v_j^{(2j-1)} = \frac{1}{j}[e,v_j^{(2j)}]$,
    \item  $\tilde{v}_{(2j-1)}^j = (-1)^j[f,\tilde{v}^j_{(2j)}]$,
    \item  $[\tilde{v}_{(2j-1)}^{j}, v_{k}^{(2k-1)}]^\sharp= \frac{-j}{k}[\tilde{v}^{j}_{(2j)},v_{k}^{(2k)}]^\sharp$.
  \end{enumerate}
\end{lemma}
\begin{proof}
(1) Let $x_j$ be the constant such that $v_j^{(2j-1)}=x_j\,[e,v_j^{(2j)}]$. Then $v_j^{(2j)}=[f,v_j^{(2j-1)}]= x_j\, [f,[e,v_j^{(2j)}]]=x_j\,[-2\fh , v_j^{(2j)}]= x_j\, j \, v_j^{(2j)}$. Hence $x_j= \frac{1}{j}.$\\
(2) Let $y_j$ be the constant such that $\tilde{v}_{(2j-1)}^j = y_j\, [f,\tilde{v}^j_{(2j)}]$. By (1), we have 
$ (v_j^{(2j-1)}|\tilde{v}^j_{(2j-1)})= \frac{y_j}{j}\big(  [e,v_j^{(2j)}]| [f,\tilde{v}^j_{(2j)}] \big)=1.$ Since 
\[ \frac{y_j}{j}\big(  [e,v_j^{(2j)}]| [f,\tilde{v}^j_{(2j)}] \big)= (-1)^j\frac{y_j}{j}([e,f]|[v_j^{(2j)}, \tilde{v}^j_{(2j)}])=(-1)^j  y_j (v_j^{(2j)}|\tilde{v}^j_{(2j)})=1,\]
we have $y_j=(-1)^j$.\\
(3) Using the above relations (1), (2) and the Jacobi identity, we have
\begin{equation} \label{eq:(lemma)N=1 bracket for F in general--}
\begin{aligned}
  & [\tilde{v}^{j}_{(2j-1)}, v_k^{(2k-1)}]^\sharp \\
  & = (-1)^j \frac1k[[f,\tilde{v}^j_{(2j)}],[e,v_k^{(2k)}]]^\sharp  = (-1)^j \frac1k[[[f,\tilde{v}^j_{(2j)}], e], v_k^{(2k)}]^\sharp \\
  & =(-1)^j \frac1k [[(-1)^{j+1}2H,\tilde{v}^j_{(2j)}], v_k^{(2k)}]^\sharp=-\frac{j}{k}[\tilde{v}^{j}_{(2j)}, v_k^{(2k)}]^\sharp.
\end{aligned}
\end{equation}
Here, the second equality in \eqref{eq:(lemma)N=1 bracket for F in general--} holds since $[e,\g]$ and  $\g^f$ intersect trivially.
\end{proof}

\begin{corollary}  \label{cor:key in lemma A.12-1}
   Let $i,j\in \mathcal{I}$ and $j>i$. Then 
  \[[\tilde{v}_{(2i-1)}^{i}, v_{j}^{(2j-1)}]^\sharp=-\frac{i}{j}[\tilde{v}_{(2i)}^{i}, v_{j}^{(2j)}]^\sharp=0\] 
  when $i$ and $j$ have the same parity or $j$ is odd.
\end{corollary}

\begin{proof}
We have $[\tilde{v}_{(2i)}^{i}, v_{j}^{(2j)}]^\sharp = (-1)^{j-i} (\tilde{v}^{j-i}_{(2j-2i)}|[\tilde{v}_{(2i)}^{i}, v_{j}^{(2j)}])v_{j-i}^{(2j-2i)}$. Since $[\tilde{v}^{j-i}_{(2j-2i)}, \tilde{v}_{(2i)}^{i}]=0$ if $i$ or $j-i$ is even, the second equality follows. The first equality follows from Lemma \ref{lem:N=1 bracket for F in general--}.
\end{proof}

Combining the result of Lemma \ref{lem:key lemma for brackets-2} and Corollary \ref{cor:key in lemma A.12-1}, we can narrow down the number of  terms in \eqref{eq:theorem, SUSY bracket}.
\begin{corollary} \label{cor:key in lemma A.12}
Let $b= F$ and $a= v_{i}^{(2i)}\in \g$ for $i\geq 5$. The longest nonzero terms in \eqref{eq:theorem, SUSY bracket} have one of the following forms:
\begin{align}
 &  \widetilde{w}(\overline{[b,v_j^{(2j-2)}]}) \widetilde{w}(\overline{[\tilde{v}^j_{(2j-1)},v_l^{(2l-1)}]})\widetilde{w}(\overline{[\tilde{v}^l_{(2l)},a]}) \label{eq: first in A.11},\\
 &  \widetilde{w}(\overline{[b,v_2^{(0)}]}) \widetilde{w}(\overline{[\tilde{v}_{(1)}^2,v_l^{(2l-1)}]})\widetilde{w}(\overline{[\tilde{v}^l_{(2l)},a]}).  \label{eq: second in A.11}
\end{align}
Moreover, $\eqref{eq: first in A.11} \neq 0$ only if $i,j,l  \in \mathcal{I}$ obey one of the three following constraints:
\begin{itemize}
\item $i=l$ is even and $j$ is odd such that $j\leq l$,
\item $i$ is even and $j=l$ is odd such that $j=l<i$,
\item $j=l=i$,
\end{itemize}
and $\eqref{eq: second in A.11}  = -\frac12 k(\chi+D)\w(\bar{v}_l^{(2l)}) \widetilde{\w}(\overline{[\tilde{v}_{(2l)}^l, v_i^{(2i)}]})\neq 0$ only if 
\begin{itemize}
  \item $i=l$,
  \item $i$ is even and $l$ is odd such that $l<i$.
\end{itemize}

\end{corollary}

\begin{lemma} \label{lem:N=1 bracket for F in general}
 Suppose  $n\geq 3$ and $2<i\leq n$. For  $b=F$ and 
 $a=v^{(4i-2)}_{2i-1}$, the followings are all nontrivial terms in \eqref{eq:theorem, SUSY bracket}:
    \begin{enumerate}
      \item   terms in relation \eqref{eq:N=1 bracket for F and v3 w(sl(3|2))};
      \item  for $j$ such that $1<j< i$, 
        \begin{align}  
        \label{eq:N=1(F,(2n-1))-1(2)}
          &\widetilde{\w}(\overline{[b,v^{(4j-2)}_{2j}]})\,\widetilde{\w}(\overline{[\tilde{v}^{2j}_{(4j-1)},a]}) 
          =-w(\bar{v}_{2j}^{(4j)}) \w(\overline{[\tilde{v}_{(4j-1)}^{2j}, a]}^\sharp) =\w(\bar{v}_{2i-2j}^{(4i-4j)})\, \w(\overline{[\tilde{v}_{(4i-4j-1)}^{2i-2j}, a]}^\sharp);  \\
          & \widetilde{\w}(\overline{[b,v^{(4j-4)}_{2j-1}]})\,\widetilde{\w}(\overline{[\tilde{v}^{2j-1}_{(4j-3)},a]}) 
           =-\w(\bar{v}_{2j-1}^{(4j-2)})\, \w(\overline{[\tilde{v}^{2j-1}_{(4j-3)}, a]}^\sharp) =\w(\bar{v}_{2i-2j+1}^{(4i-4j+2)}) \,\w(\overline{[\tilde{v}^{2i-2j+1}_{(4i-4j+1)}, a]}^\sharp). \label{eq:N=1(F,(2n-1))-2(4)}  
        \end{align}
    \end{enumerate}
  \end{lemma}
 \begin{proof}
  By Lemma \ref{lem:key lemma for brackets-2} and Corollary \ref{cor:key in lemma A.12}, we can find all possible nonzero terms. For 
  equalities \eqref{eq:N=1(F,(2n-1))-1(2)}, we use Lemma \ref{lem:N=1 bracket for F in general--} (2). Since 
  \[ [\tilde{v}_{(4j-1)}^{2j}, a]^\sharp = ([[f,\tilde{v}^{2j}_{(4j)}], a]|\tilde{v}^{2i-2j}_{(4i-4j)})\,v_{2i-2j}^{(4i-4j)} =(\tilde{v}^{2j}_{(4j)}| [\tilde{v}^{2i-2j}_{(4i-4j-1)}, a])\,v_{2i-2j}^{(4i-4j)},\]
  we have 
  \[ 
 \begin{split}
\w(\bar{v}_{2j}^{(4j)})\, \w(\overline{[\tilde{v}_{(4j-1)}^{2j}, a]}^\sharp)&= (\tilde{v}^{2j}_{(4j)}| [\tilde{v}^{2i-2j}_{(4i-4j-1)}, a])\,\w(\bar{v}_{2j}^{(4j)})\, \w(\bar{v}_{2i-2j}^{(4i-4j)})\\
&= - \w(\bar{v}_{2i-2j}^{(4i-4j)})\, \w(\overline{[\tilde{v}_{(4i-4j-1)}^{2i-2j}, a]}^\sharp). 
\end{split} \]
  The same proof works for equality \eqref{eq:N=1(F,(2n-1))-2(4)}. Hence the lemma follows.
 \end{proof}
  \begin{lemma}  \label{lem:N=1 bracket for F in general-2}
    Suppose  $n\geq 4$ and $2<i\leq n$. For  $b=F$ and 
    $a=v^{(4i)}_{2i}$,  the followings are all nontrivial terms in \eqref{eq:theorem, SUSY bracket}:
  \item  \begin{enumerate}
    \item   terms in relation \eqref{eq:N=1 bracket for F and v4 w(sl(3|2))};
    \item  for $j$ such that $1<j<i$, we have the five following possibilities:
     \begin{equation} \label{eq:N=1(F,(2n))-1}
      \begin{aligned}
        &\widetilde{\w}(\overline{[b,v^{(4j-4)}_{2j-1}]})\,\widetilde{\w}(\overline{[\tilde{v}^{2j-1}_{(4j-3)},v^{(4i-1)}_{2i}]})\,\widetilde{\w}(\overline{[\tilde{v}^{2i}_{(4i)},a]})= - \frac{2j-1}{2i}k\, \w(\bar{v}_{2j-1}^{(4j-2)})\,\w(\overline{[\tilde{v}^{2j-1}_{(4j-2)}, v_{2i}^{(4i)}]}^{\sharp})\chi ,\\
        &\widetilde{\w}(\overline{[b,v^{(4j-4)}_{2j-1}]})\,\widetilde{\w}(\overline{[\tilde{v}^{2j-1}_{(4j-3)},v^{(4j-3)}_{2j-1}]})\,\widetilde{\w}(\overline{[\tilde{v}^{2j-1}_{(4j-2)},a]})=k\, \w(\bar{v}_{2j-1}^{(4j-2)})(\chi+D)\w(\overline{[\tilde{v}^{2j-1}_{(4j-2)},v_{2i}^{(4i)}]}^{\sharp}),\\
        &\widetilde{\w}(\overline{[b,v^{(1)}_{2}]})\,\widetilde{\w}(\overline{[\tilde{v}^{2}_{(1)},v^{(4j-3)}_{2j-1}]})\,\widetilde{\w}(\overline{[\tilde{v}^{2j-1}_{(4j-2)},a]})=-\frac{1}{2}k\, (\chi+D)\w(\bar{v}_{2j-1}^{(4j-2)})\w(\overline{[\tilde{v}^{2j-1}_{(4j-2)},v_{2i}^{(4i)}]}^{\sharp});
      \end{aligned}
    \end{equation}
  \begin{equation}
  \begin{split}
        \widetilde{\w}(\overline{[b,v^{(4j-4)}_{2j-1}]})\,\widetilde{\w}(\overline{[\tilde{v}^{2j-1}_{(4j-3)},a]}) 
        &=-\w(\bar{v}^{(4j-2)}_{2j-1})\,\w(\overline{[\tilde{v}^{2j-1}_{(4j-3)},a]}^\sharp)\\
        &=\w(\bar{v}^{(4i-4j+4)}_{2i-2j+2})\,\w(\overline{[\tilde{v}^{2i-2j+2}_{(4i-4j+3)},a]}^\sharp) , \label{eq:N=1(F,(2n))-2(2)}\\
 \end{split}
\end{equation}
  \begin{equation}
  \begin{split}
        \widetilde{\w}(\overline{[b,v^{(4j-2)}_{2j}]})\,\widetilde{\w}(\overline{[\tilde{v}^{2j}_{(4j-1)},a]}) 
        &=-\w(\bar{v}^{(4j)}_{2j})\,\w(\overline{[\tilde{v}^{2j}_{(4j-1)},a]}^\sharp)\\
       & = \w(\bar{v}^{(4i-4j+2)}_{2i-2j+1})\,\w(\overline{[\tilde{v}^{2i-2j+1}_{(4i-4j+1)},a]}^\sharp). \label{eq:N=1(F,(2n))-2(4)}
 \end{split}
\end{equation}
  \end{enumerate}
\end{lemma}
 \begin{proof}
  Again, by Lemma \ref{lem:key lemma for brackets-2} and Corollary \ref{cor:key in lemma A.12}, we can find all possible nonzero terms.
 Equalities  \eqref{eq:N=1(F,(2n))-2(2)} and \eqref{eq:N=1(F,(2n))-2(4)} can be proved similarly to the proof of \eqref{eq:N=1(F,(2n-1))-1(2)}.
     The equalities in \eqref{eq:N=1(F,(2n))-1} are obtained by direct computations using Lemma \ref{lem:N=1 bracket for F in general--} (3).
 \end{proof}

\subsubsection{Proof of (2) in Proposition \ref{prop:N=1 to N=2, sl(n+1|n)}}
When $n=2$, the proof is done in Corollary \ref{eq:F and v_{2n-1} and v_{2n}}. Let us assume $n\geq 3$. For the cases, Proposition \ref{prop:N=1 to N=2, sl(n+1|n)} (2) is proved by the following lemmas.

\begin{lemma} \label{lem:final lemma for Prop N=1--1}
   For any $i\in \cI$, we have
  \[ [\w(\bar{v}^{(4i-2)}_{2i-1}){}_\Lambda \w(\bar{F})]= k^2 \bigg( i\lambda + \frac{1}{2}\chi D + \frac{2i-1}{2}\partial \bigg)\w(\bar{v}^{(4i-2)}_{2i-1}).\]
\end{lemma}

\begin{proof}
Recall Lemma \ref{lem:N=1 bracket F and tilde f},  Lemma \ref{lem:N=1 bracket for F in w(sl(3|2))} and Lemma \ref{lem:N=1 bracket for F in general}. Denote $b=F$ and $a=v_{2i-1}^{(4i-2)}$ as in Lemma \ref{lem:N=1 bracket for F in general}. It is enough to show that 
 $\sum_{j=2}^{i-1}\eqref{eq:N=1(F,(2n-1))-1(2)} + \sum_{j=2}^{i-1}\eqref{eq:N=1(F,(2n-1))-2(4)}=0$ for all $i\geq 5$. Indeed, we have 
 \begin{equation}
  \begin{aligned}
     \ 2 \ & \sum_{j=2}^{i-1}\widetilde{\w}(\overline{[b,v^{(4j-2)}_{2j}]})\widetilde{\w}(\overline{[\tilde{v}^{2j}_{(4j-1)}, a]}) \\
     = &  \sum_{j=2}^{i-1} \bigg( -w(\bar{v}_{2j}^{(4j)})\, \w(\overline{[\tilde{v}_{(4j-1)}^{2j}, a]}^\sharp)+ \w(\bar{v}_{2i-2j}^{(4i-4j)})\, \w(\overline{[\tilde{v}_{(4i-4j-1)}^{2i-2j}, a]}^\sharp) \bigg)\\
     = & - w(\bar{v}_{2i-2}^{(4i-4)})\, \w(\overline{[\tilde{v}_{(4i-5)}^{2i-2}, a]}^\sharp) +  \w(\bar{v}_{2}^{(4)})\, \w(\overline{[\tilde{v}_{(3)}^{2}, a]}^\sharp)
  \end{aligned}
 \end{equation}
 and $[\tilde{v}_{(4i-5)}^{2i-2}, a]^\sharp=[\tilde{v}_{(3)}^{2}, a]^\sharp=0$. Hence $\sum_{j=2}^{i-1}\eqref{eq:N=1(F,(2n-1))-1(2)}=0$. Similarly, we also have  $\sum_{j=2}^{i-1}\eqref{eq:N=1(F,(2n-1))-2(4)}=0$.
\end{proof}

\begin{lemma} \label{lem:final lemma for Prop N=1--2}
  For any $i\in \cI$, we have
 \[ [\w(\bar{v}^{(4i)}_{2i}){}_\Lambda \w(\bar{F})]= -k^2 \bigg( \frac{2i+1}{2}\lambda + \frac{1}{2}\chi D + i\partial \bigg)\w(\bar{v}^{(4i)}_{2i}).\]
\end{lemma}

\begin{proof}
We aim to show $\sum_{j=2}^{i-1}\eqref{eq:N=1(F,(2n))-1}=0$ and $\sum_{j=2}^{i-1}\eqref{eq:N=1(F,(2n))-2(2)}+\eqref{eq:N=1(F,(2n))-2(4)}=0$. 
The first assertion follows from the fact that 
\[ \w(\bar{v}_{2j-1}^{(4j-2)})\,\w(\overline{[\tilde{v}^{2j-1}_{(4j-2)},v_{2i}^{(4i)}]}^{\sharp})= \w(\bar{v}_{2i-2j+1}^{(4i-4j+2)})\,\w(\overline{[\tilde{v}^{2i-2j+1}_{(4i-4j+2)},v_{2i}^{(4i)}]}^{\sharp}),\]
which can be deduced by a  computation similar to the proof of \eqref{eq:N=1(F,(2n-1))-1(2)}.
For the second assertion,  we 
use equalities in \eqref{eq:N=1(F,(2n))-2(2)} and \eqref{eq:N=1(F,(2n))-2(4)}. More precisely, observe that 
\begin{equation} \label{eq:eq:N=1(F,(2n-1))-sum2}
\begin{aligned}
  & \sum_{j=2}^{i-1}\widetilde{\w}(\overline{[b,v^{(4j-2)}_{2j}]})\widetilde{\w}(\overline{[\tilde{v}^{2j}_{(4j-1)},v_{2i}^{(4i)}]})+ 
  \sum_{j=2}^{i-1}\widetilde{\w}(\overline{[b,v^{(4j-4)}_{2j-1}]})\widetilde{\w}(\overline{[\tilde{v}^{2j-1}_{(4j-3)},v_{2i}^{(4i)}]})\\
  & = -\sum_{l=3}^{2i-2}\w(\bar{v}^{(2l)}_{l})\,\w(\overline{[\tilde{v}^{l}_{(2l-1)},a]}^\sharp)= \sum_{l=3}^{2i-2}\w(\bar{v}^{(2l)}_{l})\,\w(\overline{[\tilde{v}^{l}_{(2l-1)},a]}^\sharp).
\end{aligned}
\end{equation}
Here, the first equality in \eqref{eq:eq:N=1(F,(2n-1))-sum2} corresponds to the first equality in the formulas  \eqref{eq:N=1(F,(2n))-2(2)} and  \eqref{eq:N=1(F,(2n))-2(4)}. The second equality in \eqref{eq:eq:N=1(F,(2n-1))-sum2} corresponds to the second equality of  formulas 
\eqref{eq:N=1(F,(2n))-2(2)} and  \eqref{eq:N=1(F,(2n))-2(4)}. Hence $\eqref{eq:eq:N=1(F,(2n-1))-sum2}=0$.
\end{proof}

\subsection{Proof of Theorem \ref{thm:N=1 to N=2, sl(n+1|n)}} \label{sec:proof of N=2 primary decomp}

We use Theorem \ref{thm: N=1, N=2 pva relation with superconformal} to prove Theorem \ref{thm:N=1 to N=2, sl(n+1|n)}. \\

Take $\mathcal{C}=\{v_i^{(2i)}, v_i^{(2i-1)}\,|\,i=1, \cdots 2n\}$ for $\{v_i^{(m)}\,|\,(i,m)\in {\cJ}\}$ as in section \ref{Proof of proposition}.
Let us denote $ J=\frac{\sqrt{-1}}{k}\w(\bar{\tilde{f}})$ and recall that $G= -\frac{2}{k^2} \w(\bar{F})$ is a $N=1$ superconformal vector. Then $\{ J {}_\Lambda J \}= -G+ 2k^2 \lambda \chi$ and hence $-J_{(0|0)}J=G$.  

In order to prove that $J$ is a $N=2$  superconformal vector and that all the elements of ${W}_{\osp(1|2)}^{N=2}$  are $J$-primary, 
it is enough to show that $\{J {}_\Lambda \w(\bar{v}^{(4i-2)}_{2i-1})\}= J_{(0|0)}\w(\bar{v}^{(4i-2)}_{2i-1})$ for $i=2,3,\cdots, n$ and 
that $\WW^k(\bar{\g},f)$  is freely generated by $\{\w(\bar{v}^{(4i-2)}_{2i-1}), J_{(0|0)}\w(\bar{v}^{(4i-2)}_{2i-1})|i=2,3,\cdots, n\}\cup \{J,G\}$ as a $\CC[\nabla]$-algebra.

\begin{lemma} \label{lem:N=1 bracket tilde(f)}
  For $i=2,  \cdots, n$, we have 
  \[  \{   \w(\bar{\tilde{f}}){}_\Lambda\w(\bar{v}^{(4i-2)}_{2i-1})\}= -\w(\bar{v}^{(4i)}_{2i})+ \sum_{j=2}^{i-1} c_{j,i} \w(\bar{v}^{(4j)}_{2j}) \w(\bar{v}^{(4i-4j)}_{2i-2j}),\]
  for some constants $c_{j,i} \in \CC$. 
\end{lemma}
\begin{proof}
Using relation \eqref{eq:theorem, SUSY bracket} and Lemma \ref{eq:key lemma for brackets}, we have 
\begin{equation}
  \begin{aligned}
  \{ \w(\bar{v}^{(4i-2)}_{2i-1}) {}_\Lambda \w(\bar{\tilde{f}})\} & =      -\w(\bar{v}^{(4i)}_{2i})+\sum_{j=2}^{i-1}  \w(\overline{[\tilde{f}, v_{2j}^{(4j-2)}]})\w(\overline{[\tilde{v}_{(4j-1)}^{2j}, v^{(4i-2)}_{2i-1}]^\sharp})\\
  & =-\w(\bar{v}_{2i}^{(4i)})- \sum_{j=2}^{i-1}\w(\bar{v}^{(4j)}_{2j}) \big(\tilde{v}^{2i-2j}_{(4i-4j)}| [\tilde{v}^{2j}_{(4i-1)}, v^{(4i-2)}_{2i-1}] \big)v^{(4i-4j)}_{2i-2j}.
  \end{aligned}
\end{equation}
We let $c_{j,i}=-\big(\tilde{v}^{2i-2j}_{(4i-4j)}| [\tilde{v}^{2j}_{(4j)}, v^{(4i-2)}_{2i-1}] \big)$ and get the lemma by skewsymmetry.
\end{proof}

By Lemma \ref{lem:N=1 bracket tilde(f)}, 
we conclude that $J$ is a $N=2$ superconformal vector. Moreover, the set $W_{\osp(1|2)}^{N=2}=\{\w(\bar{v}^{(4i-2)}_{2i-1})|i=2,3,\cdots, n\}$ satisfies the second statement of Theorem \ref{thm:N=1 to N=2, sl(n+1|n)}.

\subsection{Proof of Theorem \ref{thm:N=2 conformal,N=0 W-alg}} \label{sec:Proof of theorems N=0}

We use Theorem \ref{thm: pva, N=1 pva relation with superconformal} and \ref{thm: N=1, N=2 pva relation with superconformal} to prove Theorem \ref{thm:N=2 conformal,N=0 W-alg}.

\begin{lemma} \label{lem:first lemma for N=1 scv in N=2 W}
  The following  relations hold:
  \begin{equation}\label{eq:first lemma for N=1 scv in N=2 W}
  \begin{split}
&   \{ \nu(f){}_\lambda \nu(f)\}= -2\nu(F)+\frac{1}{2}(\nu(U))^2-k^2 \lambda^2 , \\
& \{ \nu(\tilde{f}){}_\lambda \nu(\tilde{f})\}= 2\nu(F)-\frac{1}{2}(\nu(U))^2+k^2 \lambda^2.
  \end{split}
  \end{equation}
  Hence if we consider $G= \sqrt{-\frac{1}{k}}\nu (f)$ or $\sqrt{\frac{1}{k}}\nu (\tilde{f})$ then $G_{(0)}G= \frac{2}{k}\left(\nu(F)-\frac{1}{4}(\nu(U))^2\right)$, and $G$ is a conformal vector of $\WW^k(\g, F)$.
\end{lemma}
\begin{proof}
  Applying Theorem \ref{thm:nonSUSY w bracket}, we get the relations \eqref{eq:first lemma for N=1 scv in N=2 W}. 
  By Theorem \ref{thm:conformal in N=0 W}, the last assertion directly follows from equations \eqref{eq:first lemma for N=1 scv in N=2 W}.
\end{proof}

Let $G= \sqrt{-\frac{1}{k}}\nu (f)$ and  $L:=G_{(0)}G= \frac{2}{k}(\nu(F)-\frac{1}{4}(\nu(U))^2)$. Then, for $s=2i,2i-1$ and $i\neq 2$, the elements $\nu(v_i^{(s)})$ in $W_{\sll_2}^{N=0}$ and $G$ are $L$-primary by Theorem \ref{thm:conformal in N=0 W}. We also have $\{G_\lambda \nu(v_{i}^{(2i-1)})\}= G_{(0)} \nu(v_{i}^{(2i-1)})$ for $i\neq 2$
since 
\begin{equation} \label{eq: N=1 scv in N=2 W (1)}
  \{\nu(f){}_\lambda \nu(v_{i}^{(2i-1)})\}= \nu(v^{(2i)}_{i})-\sum_{j=2}^{i-1}(-1)^{j}\nu([f, v_{j}^{(2j-2)}])\nu([\tilde{v}^{j}_{(2j)}, v_{i}^{(2i-1)}]^{\natural}),
\end{equation}
and thus $\{G, L\}\cup \{\nu(v_{i}^{(2i-1)}), G_{(0)}\nu(v_{i}^{(2i-1)})|i=1,3,4,\cdots, 2n \}$ freely generates $\WW^k(\g,F)$ as a $\CC[\partial]$-algebra. For $D:=G_{(0)}$,
 by Lemma \ref{lem:first lemma for N=1 scv in N=2 W} and  Theorem \ref{thm: pva, N=1 pva relation with superconformal}, we conclude that (i) the element $G$ is a $N=1$ superconformal vector, (ii) $\nu(v_{i}^{(2i-1)})$ for $i\neq 2$ is $G$-primary, (iii) the set $\{\nu(v_{i}^{(2i-1)})|i=1,2,\cdots, 2n\}$  freely generates   $\WW^k(\g,F)$ as a $\CC[\nabla]$-algebra.

Similarly, we can show $\tilde{G}=\sqrt{\frac{1}{k}}\nu (\tilde{f})$ is another $N=1$ superconformal vector.
Since 
\begin{equation} \label{eq: N=1 scv in N=2 W (2)}
  \{\nu(\tilde{f}){}_\lambda \nu(v^{(4i'-2)}_{2i'-1})\}= \nu(v^{(4i')}_{2i'}),\quad  \{\nu(\tilde{f}){}_\lambda \nu(v_{2i-1}^{(4i-3)})\}= -\nu(v_{2i}^{(4i-1)}), 
\end{equation}
for $i=1,2, \cdots, n$ and $i'=2,3, \cdots,n$, 
the set \[\{\tilde{G}, L\}\cup \{\nu(v^{(4i'-2)}_{2i'-1}), \tilde{G}_{(0)}\nu(v^{(4i'-3)}_{2i'-1})\}_{i'=2}^n \cup \{\nu(v_{2i-1}^{(4i-3)}), \tilde{G}_{(0)}\nu(v_{2i-1}^{(4i-3)})\}_{i=1}^n\] freely generates $\WW^k(\g,F)$ as a $\CC[\partial]$-algebra. If we denote $\tilde{D}:=\tilde{G}_{(0)}$ then (i) $\tilde{G}$ is a $N=1$ superconformal vector (ii) $\nu(v^{(4i'-2)}_{2i'-1})$ and $\nu(v_{2i-1}^{(4i-3)})$ are $\tilde{G}$-primary (iii)
 $\{\tilde{G}\}\cup \{\nu(v^{(4i'-2)}_{2i'-1}), \nu(v_{2i-1}^{(4i-3)})|i'=2,3, \cdots, n \text{ and } \, i=1,2,\cdots, n \}$ freely generates $\WW^k(\g, F)$ as a $\CC[\tilde{\nabla}]$-algebra.
 
In the rest of this section, we show that $J:=-\sqrt{-1}\,\nu(U)$ is a $N=2$  superconformal vector of $\WW^k(\g,F)$. Let us first show that the odd derivations
$ D$ and $\tilde{D}$ 
give rise to $N=2$ SUSY structure on $\WW^k(\g,F)$. Observe that 
\begin{equation*}
\begin{split}
    &DJ= \sqrt{-\frac{1}{k}}\nu(f)_{(0)}\,J=\sqrt{\frac{1}{k}}\nu(\tilde{f})=\tilde{G}, \\
    &  \tilde{D}J= \sqrt{\frac{1}{k}}\nu(\tilde{f})_{(0)}\, J=-\sqrt{-\frac{1}{k}}\nu(f)=-G.
\end{split}
\end{equation*}
Now, using relations \eqref{eq: N=1 scv in N=2 W (1)} and \eqref{eq: N=1 scv in N=2 W (2)}, we get 
\begin{equation}
  \begin{aligned}
    & \big\{DJ{}_\lambda \nu(v_{2i-1}^{(4i-3)})\big\}= \sqrt{\frac{1}{k}}\nu(\tilde{f})_{(0)}\nu(v_{2i-1}^{(4i-3)})=\tilde{D}\,\nu(v_{2i-1}^{(4i-3)}),\\
    & \big\{\tilde{D}J{}_\lambda \nu(v_{2i-1}^{(4i-3)})\big\}=-\sqrt{-\frac{1}{k}}\nu(v_{2i-1}^{(4i-3)})=-D\, \nu(v_{2i-1}^{(4i-3)}).\\
  \end{aligned}
\end{equation}
and
\begin{equation}
D\tilde{D}J= G_{(0)}(-G)=\sqrt{-\frac{1}{k}} \nu(f)_{(0)}\bigg(-\sqrt{\frac{-1}{k}}\nu(f)\bigg)=-L.
\end{equation}
We can also check $\{J{}_\lambda J\}=2k\lambda$ and $\big\{J{}_\lambda\nu(v_{2i-1}^{(4i-3)})\big\}=0$. Therefore,
from Theorem \ref{thm:N=2 superconformal in N=0 PVA}, we deduce that $J$ is a $N=2$ superconformal vector and $\nu(v_{2i-1}^{(4i-3)})$ for $i=2,\cdots, n$ are $J$-primary. Moreover, relations \eqref{eq: N=1 scv in N=2 W (1)} and \eqref{eq: N=1 scv in N=2 W (2)} show that $\{J\}\cup \{\nu(v_{2i-1}^{(4i-3)})|i=2,\cdots, n\}$ freely generates $\WW^k(\g,F)$ as a $\CC[\bm{\nabla}]$-algebra.


\begin{thebibliography}{10}

\bibitem{Ademollo}
M.~Ademollo, L.~Brink, A.~D'Adda, R.~D'Auria, E.~Napolitano, S.~Sciuto, E.~{Del
  Giudice}, P.~{Di Vecchia}, S.~Ferrara, F.~Gliozzi, R.~Musto, and
  R.~Pettorino.
\newblock Supersymmetric strings and colour confinement.
\newblock {\em Physics Letters B}, 62(1):105--110, 1976.

\bibitem{BDK09}
A.~Barakat, A.~De~Sole, V.G.~Kac.
\newblock
Poisson vertex algebras in the theory of Hamiltonian equations. 
\newblock{\em
Jpn. J. Math.},
4:141--252, 2009.




\bibitem{Bar00}
K.~Barron.
\newblock{$N=1$  Neveu-Schwarz vertex operator superalgebras over Grassmann algebras and with odd formal variables}
\newblock{\em
China Higher Education Press (CHEP)}, Beijing:9--35, 2000.


\bibitem{BPZ}
A.~A. Belavin, A.~M. Polyakov, and A.~B. Zamolodchikov.
\newblock Infinite conformal symmetry in two-dimensional quantum field theory.
\newblock {\em Nuclear Phys. B}, 241(2):333--380, 1984.

\bibitem{BHS08}
D.~Ben-Zvi, R.~Heluani,  M.~Szczesny.
\newblock Supersymmetry of the chiral de Rham complex.
\newblock{\em Compos. Math.}, 144 (2):503--521, 2008.



\bibitem{BMN}
D.~Berenstein, J.~Maldacena, and H.~Nastase.
\newblock Strings in flat space and pp waves from {$\mathscr{N}=4$} super
  {Y}ang {M}ills.
\newblock {\em J. High Energy Phys.}, 2002(4):No. 13, 30, 2002.

\bibitem{BDKRSVV}
Z.~Bern, L.~J. Dixon, D.~A. Kosower, R.~Roiban, M.~Spradlin, C.~Vergu, and
  A.~Volovich.
\newblock Two-loop six-gluon maximally helicity violating amplitude in
  maximally supersymmetric {Y}ang-{M}ills theory.
\newblock {\em Phys. Rev. D}, 78(4):045007, 25, 2008.

\bibitem{BLN93} M.~Bershadsky, W.~Lerche, D.~Nemeschansky, N.~Warner. 
\newblock  Extended N=2 superconformal structure of gravity and W-gravity coupled to matter.
\newblock {\em Nucl. Phys. B} 401(1-2):304--347, 1993.

\bibitem{Borch}
R.~E. Borcherds.
\newblock Vertex algebras, {K}ac-{M}oody algebras, and the {M}onster.
\newblock {\em Proc. Nat. Acad. Sci. U.S.A.}, 83(10):3068--3071, 1986.

\bibitem{BS}
P.~Bouwknegt and K.~Schoutens.
\newblock {$\mathscr W$} symmetry in conformal field theory.
\newblock {\em Phys. Rep.}, 223(4):183--276, 1993.

\bibitem{CHSW}
P.~Candelas, G.~T. Horowitz, A.~Strominger, and E.~Witten.
\newblock Vacuum configurations for superstrings.
\newblock {\em Nuclear Phys. B}, 258(1):46--74, 1985.

\bibitem{CS}
S.~Carpentier and U.~R. Suh.
\newblock Supersymmetric bi-{H}amiltonian systems.
\newblock {\em Comm. Math. Phys.}, 382(1):317--350, 2021.

\bibitem{DG97} F.~Delduc, L.~Gallot,
    \newblock  KP and KdV hierarchies in extended superspace.
\newblock {\em Commun. Math. Phys.} 190(2):395--410, 1997. 

\bibitem{DRS92}
F.~Delduc, E.~Ragoucy, and P.~Sorba.
\newblock Super-{T}oda theories and {$W$}-algebras from superspace
  {W}ess-{Z}umino-{W}itten models.
\newblock {\em Comm. Math. Phys.}, 146(2):403--426, 1992.

\bibitem{delius}
G.~W. Delius.
\newblock The {$N=2$} super-{K}ac-{M}oody algebra and the {WZW}-model in
  {$(2,0)$} superspace.
\newblock {\em Internat. J. Modern Phys. A}, 5(24):4753--4767, 1990.

\bibitem{DHKS}
J.~M. Drummond, J.~Henn, G.~P. Korchemsky, and E.~Sokatchev.
\newblock Dual superconformal symmetry of scattering amplitudes in
  {${\mathscr{N}}=4$} super-{Y}ang-{M}ills theory.
\newblock {\em Nuclear Phys. B}, 828(1-2):317--374, 2010.

\bibitem{FOFR}
J.~M. Figueroa-O'Farrill and E.~Ramos.
\newblock Classical {${N}=2$} {$W$}-superalgebras and supersymmetric
  {G}el\textquotesingle fand-{D}ickey brackets.
\newblock {\em Nuclear Phys. B}, 368(2):361--376, 1992.

\bibitem{dico}
L.~Frappat, A.~Sciarrino, and P.~Sorba.
\newblock  Dictionary on {L}ie algebras and superalgebras.
\newblock Academic Press, Inc., San Diego, CA, 2000.
\newblock With 1 CD-ROM (Windows, Macintosh and UNIX).

\bibitem{FMS}
D.~Friedan, E.~Martinec, and S.~Shenker.
\newblock Conformal invariance, supersymmetry and string theory.
\newblock {\em Nuclear Phys. B}, 271(1):93--165, 1986.

\bibitem{gaiotto}
D.~Gaiotto and E.~Witten.
\newblock Supersymmetric boundary conditions in {$\mathscr{N}=4$} super
  {Y}ang-{M}ills theory.
\newblock {\em J. Stat. Phys.}, 135(5-6):789--855, 2009.

\bibitem{GS92} B.~Gato-Rivera, A.M.~Semikhatov, 
 \newblock  $d\leq1\bigcup d\geq25$ and constrained KP hierarchy from BRST invariance in the $c\neq3$ topological algebra,
 \newblock {\em Phys. Lett. B} 293(1-2):72--80, 1992.

\bibitem{GW}
D.~Gepner and E.~Witten.
\newblock String theory on group manifolds.
\newblock {\em Nuclear Phys. B}, 278(3):493--549, 1986.

\bibitem{HW}
A.~Hanany and E.~Witten.
\newblock Type {IIB} superstrings, {BPS} monopoles, and three-dimensional gauge
  dynamics.
\newblock {\em Nuclear Phys. B}, 492(1-2):152--190, 1997.

\bibitem{Hel07}
R.~Heluani.
\newblock
SUSY vertex algebras and supercurves.
\newblock{\em Comm. Math. Phys.}, 275(3):607–-658, 2007.


\bibitem{HK07}
R.~Heluani and V.~G. Kac.
\newblock Supersymmetric vertex algebras.
\newblock {\em Comm. Math. Phys.}, 271(1):103--178, 2007.

\bibitem{Hoyt12}
C.~Hoyt.
\newblock Good gradings of basic {L}ie superalgebras.
\newblock {\em Israel J. Math.}, 192(1):251--280, 2012.

\bibitem{IMY}
T.~Inami, Y.~Matsuo, and I.~Yamanaka.
\newblock Extended conformal algebras with {$N=1$} supersymmetry.
\newblock {\em Phys. Lett. B}, 215(4):701--705, 1988.

\bibitem{Kac98}
V.~G. Kac.
\newblock  Vertex algebras for beginners, volume~10 of {\em University
  Lecture Series}.
\newblock American Mathematical Society, Providence, RI, second edition, 1998.

\bibitem{KRW04}
V.~G. Kac, S.~Roan, and M.~Wakimoto.
\newblock Quantum reduction for affine superalgebras.
\newblock {\em Comm. Math. Phys.}, 241(2-3):307--342, 2003.

\bibitem{KZ}
V.~G. Knizhnik and A.~B. Zamolodchikov.
\newblock Current algebra and {W}ess-{Z}umino model in two dimensions.
\newblock {\em Nuclear Phys. B}, 247(1):83--103, 1984.

\bibitem{LM88} C. A.~Laberge, P.~Mathieu, 
 \newblock  $N=2$ superconformal slgebra and integrable $O(2)$ fermionic extensions of the Korteweg--de Vries equation.
 \newblock {\em Phys. Lett. B}, 215(4):718-722, 1988.

\bibitem{MadR}
J.~O. Madsen and E.~Ragoucy.
\newblock Quantum {H}amiltonian reduction in superspace formalism.
\newblock {\em Nuclear Phys. B}, 429(2):277--290, 1994.

\bibitem{MRS21}
A.~Molev, E.~Ragoucy, and U.~R. Suh.
\newblock Supersymmetric {$W$}-algebras.
\newblock {\em Lett. Math. Phys.}, 111(1):Paper No. 6, 25, 2021.

\bibitem{work2}
M.~Park.
\newblock Classical {W}-algebras associated to lie superalgebras.
\newblock {\em Master thesis in Seoul National University}.

\bibitem{Pope}
C.~N. Pope.
\newblock Review of {W} strings.
\newblock In {\em Proc. of Int. Symp. on Black Holes, Worm Holes, Membranes and
  Superstrings}. Woodlands TX, 1992.

\bibitem{RSS96}
E.~Ragoucy, A.~Sevrin, and P.~Sorba.
\newblock Strings from ${N}=2$ gauged {W}ess-{Z}umino-{W}itten models.
\newblock {\em Comm. Math. Phys.}, 181(1):91--129, 1996.

\bibitem{work}
E.~Ragoucy, A.~Song, and U.~R. Suh.
\newblock Work in progress.

\bibitem{RSS23}
E.~Ragoucy, A.~Song, and U.~R. Suh.
\newblock Generators of supersymmetric classical {$W$}-algebras.
\newblock {\em Comm. Math. Phys.}, 397(1):111--139, 2023.

\bibitem{RS90}
E.~Ragoucy and P.~Sorba.
\newblock Super-{K}ac-{M}oody algebras and the {$N=2$} superconformal case.
\newblock {\em Phys. Lett. B}, 245(3-4):465--470, 1990.

\bibitem{Suh20}
U.~R. Suh.
\newblock Structures of (supersymmetric) classical {W}-algebras.
\newblock {\em J. Math. Phys.}, 61(11):111701, 27, 2020.

\bibitem{Wak99}
M.~Wakimoto.
\newblock {\em Infinite-dimensional {L}ie algebras}, volume 195 of {\em
  Translations of Mathematical Monographs}.
\newblock American Mathematical Society, Providence, RI, 2001.
\newblock Translated from the 1999 Japanese original by Kenji Iohara, Iwanami
  Series in Modern Mathematics.

\bibitem{W84}
E.~Witten.
\newblock Nonabelian bosonization in two dimensions.
\newblock {\em Comm. Math. Phys.}, 92(4):455--472, 1984.

\bibitem{Za}
A.~B. Zamolodchikov.
\newblock Infinite extra symmetries in two-dimensional conformal quantum field
  theory.
\newblock {\em Teoret. Mat. Fiz.}, 65(3):347--359, 1985.




 







\end{thebibliography}
\end{document}